\documentclass[11pt]{article}
\usepackage[margin=0.85in]{geometry}
\usepackage{amsmath,amssymb,amsfonts,amsthm,mathtools}
\usepackage{booktabs,longtable,array,ragged2e}
\usepackage{enumitem}
\usepackage{xcolor}
\usepackage{multirow}
\usepackage{parskip}
\usepackage{microtype,bm}
\usepackage{pdflscape}
\usepackage{float}
\usepackage{tabularx}
\usepackage{caption}
\usepackage{authblk}
\usepackage{eso-pic}
\usepackage{pdflscape}
\usepackage[hidelinks]{hyperref}
\usepackage[T1]{fontenc}
\usepackage[utf8]{inputenc}
\usepackage{parskip}
\usepackage{cancel}
\usepackage{comment}
\usepackage{cite}
\usepackage{empheq}
\usepackage[normalem]{ulem}
\makeatletter
\g@addto@macro\UrlBigBreaks{\do\/}
\makeatother

\usepackage[table]{xcolor}
\usepackage{lmodern}

\usepackage{mathrsfs}

\renewcommand{\arraystretch}{1.15}
\newcolumntype{L}[1]{>{\RaggedRight\arraybackslash}p{#1}}
\newcolumntype{C}[1]{>{\Centering\arraybackslash}p{#1}}

\hypersetup{colorlinks=true,linkcolor=blue!60!black,citecolor=blue!60!black,urlcolor=blue!60!black}

\definecolor{darkred}{rgb}{0.65,0.15,0}
\hypersetup{pdfborder={0 0 0},colorlinks=true,urlcolor=blue,citecolor=blue,linkcolor=darkred,linktocpage=true}

\newcommand{\tr}{\operatorname{tr}}

\newcommand{\three}{\mathbf 3}
\newcommand{\thbar}{\overline{\mathbf 3}}
\newcommand{\six}{\mathbf 6}
\newcommand{\sixbar}{\overline{\mathbf 6}}
\newcommand{\eight}{\mathbf 8}

\newcommand{\eq}[1]{(\ref{#1})}
\newcommand{\w}[1]{\\[0.#1cm]}
\def\be{\begin{equation}}
\def\ee{\end{equation}}
\def\bea{\begin{eqnarray}}
\def\eea{\end{eqnarray}}
\def\nn{\nonumber}
\def\m{\mu}
\def\n{\nu}

\newcommand{\trtwo}{\operatorname{tr}_{\mathbf 2}}
\newcommand{\trfour}{\operatorname{tr}_{\mathbf 4}}

\newcommand{\trH}{\operatorname{tr}_{H}}
\newcommand{\Sp}{\operatorname{Sp}}

\newcommand{\one}{\mathbf 1}

\newtheorem{lemma}{Lemma}

\newcommand{\bA}{b_{A_2}}
\newcommand{\bU}{b_{U(1)}}
\numberwithin{equation}{section}

\title{Anomaly-free six-dimensional \(Sp(1)_R\)-gauged supergravities}
\author{ Katrin Becker, Ergin Sezgin and Qi You 
\\
George P. and Cynthia W. Mitchell Institute \\
for Fundamental Physics and Astronomy \\
Texas A\&M University, College Station, TX 77843-4242, USA}

\date{}

\begin{document}

\AddToShipoutPictureFG*{%
  \AtPageUpperLeft{%
    \put(\LenToUnit{\dimexpr\paperwidth-1in\relax},
         \LenToUnit{-0.55in}){%
      \makebox(0,0)[rt]{%
        \begin{tabular}{r}
        MI-HET-889
        \end{tabular}%
      }%
    }%
  }%
}

\thispagestyle{empty}

\maketitle

\begin{abstract} 

We study anomaly constraints in six-dimensional \(N=(1,0)\) supergravity with gauged \(Sp(1)_R\) symmetry, one tensor multiplet, \(n_V\) vector multiplets and \(n_H\) hypermultiplets. Imposing cancellation of local gravitational, gauge and mixed anomalies, together with the embedding of the anomaly coefficients in a unimodular rank-two string charge lattice, the number of vector multiplets is sharply constrained: the only allowed positive values are $n_V=12$ and $n_V=96$. For a neutral abelian factor, or drone, the anomaly equations force the corresponding row and column of the abelian gauge-kinetic matrix to vanish. Such a factor therefore cannot be retained as an ordinary propagating vector multiplet. For $n_V=96$, within the class of semisimple non-\(R\) gauge algebras containing no \(A_1\) or \(A_2\) factor, the exhaustive spectra-assembly calculation yields no drone-free spectrum.
For \(n_V=12\), the semisimple non-\(R\) branch is necessarily \(A_1^3\). In the doublet/adjoint truncation described in the text, we find $4{,}337{,}331$ consistent models, classified by $102$ anomaly-vector classes. We also consider the reductive branch $A_2\times U(1)\times Sp(1)_R$. For \(A_2\) representations in $\{\mathbf 1,\mathbf 3,\overline{\mathbf 3}, \mathbf 6,\overline{\mathbf 6},\mathbf 8\}$ and primitive \(U(1)\) charge assignments with \(|q|\le 2\), we find \(178{,}538\) solutions, classified by $21$ sets of anomaly vectors.

\end{abstract}

\newpage

\tableofcontents

\section{Introduction}

Six-dimensional \(N=(1,0)\) supergravity provides a useful laboratory for studying the interplay between anomaly cancellation, supersymmetry, and quantum consistency conditions in theories with gravity.  The anomaly polynomial of chiral six-dimensional theories is highly constrained, and the Green--Schwarz mechanism ties local anomaly factorization to the geometry of the string charge lattice.  In particular, for theories with one tensor multiplet, the anomaly coefficients must be embedded in a rank-two lattice of signature \((1,1)\), which is expected to be unimodular.  This gives a powerful arithmetic refinement of the usual local anomaly-cancellation conditions.

In this paper we focus on six-dimensional \(N=(1,0)\) models in which the \(R\)-symmetry group \(Sp(1)_R\) is gauged.  Such theories are natural from the point of view of gauged supergravity and provide an interesting arena for Swampland-type questions, since consistency requires the simultaneous cancellation of gravitational, gauge, mixed gauge-gravitational, and \(R\)-symmetry anomalies.  While \(U(1)_R\)-gauged six-dimensional models have been studied in a number of contexts \cite{Randjbar-Daemi:1985tdc,Avramis:2005qt,Avramis:2005hc,Becker:2023zyb,Becker:2025xgy,Suzuki:2005vu,Becker:2026etl,Guo:2026dvy}, the corresponding \(Sp(1)_R\)-gauged models have received comparatively less systematic attention \cite{Avramis:2005hc,Suzuki:2005vu}.

Previous analyses of \(Sp(1)_R\)-gauged models include limited anomaly searches and constraints on the number of vector multiplets.  One subtlety, which is consequential for the classification problem, is the normalization of the \(Sp(1)_R\) contribution to the anomaly polynomial.  We revisit this point carefully and impose the independent quartic gauge-anomaly constraints, the rank condition required for Green--Schwarz factorization, hypermultiplet multiplicities satisfying the representation-reality conditions, and embedding of the resulting anomaly coefficients into a unimodular rank-two lattice.
A striking consequence of the unimodularity condition applied to the pair \((a,b_R)\) is that the number of vector multiplets is not arbitrary.  In the normalization used here, one obtains only two positive possibilities, $n_V=12$ and $n_V=96$ \cite{Becker:2025xgy}. This is much more restrictive than the familiar congruence condition \cite{Becker:2025xgy} appearing in \(U(1)_R\)-gauged models. This remarkable restriction on the dimension of the full gauge group is one of the main motivations for this work, in which we explore the extent to which such locally and globally anomaly-free models can be classified. 

Before conducting our search for anomaly-free $Sp(1)_R$ gauged models, we first examine whether the gauge group can contain a \(U(1)\) factor under which all matter is neutral, which we call a drone. In line with the obstruction discussed in \cite{Park:2011wv}, we show that, for one tensor multiplet, the abelian anomaly equations imply \(b_{DD}=b_{DI}=0\) for every abelian index \(I\), where \(D\) labels the drone. The \(D\)-th row and column of the gauge-kinetic matrix \(j\cdot b_{IJ}\) therefore vanish, so the matrix cannot be positive definite. We allow drones as auxiliary factors in the spectra-assembly procedure but exclude them from the final physical models.

We first analyze the \(n_V=96\) branch. In order to keep the classification finite and computationally manageable, we adopt the simplifying assumption that the non-\(R\) gauge algebra contains no $U(1), A_1, A_2$ factors. This is the only structural exclusion imposed on the non-\(R\) simple factors in the \(n_V=96\) scan. Using a method we refer to as the spectra-assembly method, which has similarities to various graphical methods that have been employed previously, we find that the solution space is empty. As a by-product, we find 4,094 formal anomaly- and lattice-consistent spectra; every one of these spectra, however, necessarily contains one or more drones. We therefore rule them out as physical models because the corresponding abelian gauge-kinetic matrix is degenerate. 

Next, we turn to the \(n_V=12\) branch.  If the non-\(R\) gauge group is required to be semisimple, the dimension condition forces $ G_{{\rm non\mbox{-}}R}=A_1^3$. 
In the doublet/adjoint truncation, we allow unrestricted multiplicities of the permitted blocks, subject only to the requirement that an individual product representation contain a triplet under at most one $A_1$ factor. We find $4{,}337{,}331$ consistent models, classified by $102$ sets of anomaly vectors. We provide a representative spectrum from each class.  We also study the reductive branch $A_2\times U(1)\times Sp(1)_R$, 
where the classification becomes more intricate because one must also enumerate \(U(1)\) charge assignments.  For \(A_2\) representations in
$
\{\mathbf 1,\mathbf 3,\overline{\mathbf 3},\mathbf 6,\overline{\mathbf 6},\mathbf 8\}
$
and primitive charge assignments with \(|q|\le2\), we find \(178{,}538\) solutions grouped into $21$ classes.

The paper is organized as follows. In section~\ref{sec:anomaly-polynomial} we review the anomaly polynomial for six-dimensional \(N=(1,0)\) supergravity with gauged \(Sp(1)_R\) symmetry. In section~\ref{sec: consistency conditions} we summarize the consistency conditions including the factorization of the anomaly polynomial, the embedding of all anomaly vectors in a unimodular rank-two lattice of signature $(1,1)$, the requirement of nonvanishing and ghost-free gauge kinetic terms. After describing the strategy for the search for consistent models in section~\ref{sec: a search strategy}, we treat the $n_V=96$ branch in section~\ref{sec: nV=96 branch} and the $n_V=12$ branch in sections~\ref{sec:nV12} and~\ref{sec:a2u1}. We conclude in section~\ref{sec:conclusions} with a summary of the main results and some
open questions.  The appendices summarize our hypermultiplet multiplicity conventions, illustrate the spectra-assembly method, and tabulate the \(n_V=12\) results. The full catalogue of formal anomaly- and lattice-consistent spectra for the \(n_V=96\) branch is supplied as ancillary material \cite{github}.

\section{The anomaly polynomial}
\label{sec:anomaly-polynomial}

We shall study the anomalies in $Sp(1)_R$ gauged $N=(1,0)$ supergravity in six dimensions coupled to $n_T=1$ tensor, $n_V$ Yang-Mills and $n_H$ hypermultiplets. In the conventions of \cite{Suzuki:2005vu} which we shall use, the model contains the multiplets
\be
\underbrace{ \left(e_\m^m, \psi_\m^{A-}, B_{\m\n}^{-} \right)}_{\rm graviton}\ ,\qquad \underbrace{\left(B_{\m\n}^{+}, \chi^{A+},\varphi  \right)}_{\rm tensor}\ ,\qquad \underbrace{\left(A_\m, \lambda^{A-} \right)}_{\rm vector}\ ,\qquad \underbrace{\left(4\phi,\psi^{+} \right)}_{\rm hyper}\ .
\ee
The two-form potentials, $B_{\m\n}^{\pm}$, have (anti-)self-dual field strengths. The spinors are symplectic Majorana-Weyl (SMW), $A=1,2$ labels the doublet components of  $R$ symmetry group $Sp(1)_R$, and  chiralities of the fermions are denoted by $\pm$. The hyperini are inert under $Sp(1)_R$ and they are typically SMW as well. 

As we shall review in the next section, the unimodularity of the rank-two charge lattice restricts the \(Sp(1)_R\)-gauged models to the two vector-multiplet numbers \cite{Becker:2025xgy}
\be
n_V=12 \qquad\text{or}\qquad n_V=96\ .
\label{1296}
\ee
In this paper, we will consider the following gauge groups
\begin{align}
n_V=96: \qquad & \qquad  \prod_{i=1}^n G_i \times Sp(1)_R\ ,\quad \sum_{i=1}^n {\rm dim}\, G_i =93\ ,
\label{96}\\
n_V=12:\qquad  & \qquad  
 \begin{cases} A_1^3 \times Sp(1)_R \\ A_2 \times U(1) \times Sp(1)_R
\end{cases}
\end{align}
where $G_i$ are simple and simply connected. To begin with, we shall write down the anomaly formula for a semisimple group
\be
G_{ss}= \prod_{i=1}^n G_i \times Sp(1)_R\\ ,
\label{g1}
\ee
and later  consider the inclusion of $U(1)$. The resulting anomaly polynomial for this model is 
known \cite{Suzuki:2005vu,Avramis:2005hc,Becker:2023zyb}. We shall briefly rederive it here because critical  factors of 2 due to the presence of SMW spinors and the trace notations can be confusing especially 
for the pseudo-real representations of $Sp(1)_R$. The \(Sp(1)_R\) gaugini decompose as \(\mathbf2\otimes\mathbf3=\mathbf2\oplus\mathbf4\). Their doublet and quartet components are denoted by \(\lambda^{A-}\) and \(\lambda^{ABC-}\), respectively, and both satisfy the SMW condition. The contributions of 
$\lambda^{A-}$ and the dilatino $\chi^{A+}$ to the anomaly polynomial precisely cancel, and we are 
left with the contributions of $\{\psi_\mu^{A-}, \lambda^{ABC-},\lambda_i^{A-},\psi^+\}$ \cite{Alvarez-Gaume:1983ihn,Alvarez-Gaume:1984zlq}, 
\begin{align}
P(\psi_\mu^{A-})
&= -\frac12\Bigg[
2\times \Big(\frac{245}{360}\tr R^4 -\frac{43}{288}(\tr R^2)^2\Big) 
+\frac{19}{6}\tr R^2\,\trtwo F_R^2 +\frac{10}{3}\trtwo F_R^4 \Bigg]\ ,
\label{eq:anom-gravitino}\w2
P(\lambda^{ABC-}) &= -\frac12\Bigg[
4\times \Big( \frac{1}{360}\tr R^4 + \frac{1}{288}(\tr R^2)^2
\Big) - \frac16\tr R^2\,\trfour F_R^2 +\frac23\trfour F_R^4 \Bigg]\ ,
\label{eq:anom-Rgaugino}\w2
P(\lambda_i^{A-}) &= -\frac12\Bigg[ 2\times (n_V-3)
\Big( \frac{1}{360}\tr R^4 + \frac{1}{288}(\tr R^2)^2 \Big) -\frac16\tr R^2\,
\tr\!\Big(F_{\rm adj}\times \mathbf 1_2+\mathbf 1_{n_V-3}\times F_R\Big)^2
\notag\\
&\qquad
+\frac23 \tr\!\Big(F_{\rm adj}\times \mathbf 1_2+\mathbf 1_{n_V-3}\times F_R\Big)^4
\Bigg] \ ,
\label{eq:anom-gaugino}\w2
P(\psi^+) &= n_H \times \left(\frac{1}{360}\tr R^4 + \frac{1}{288}(\tr R^2)^2
\right) -\frac16 \tr R^2\,\trH F^2  +\frac23 \trH F^4\ ,
\label{eq:anom-hyperino}
\end{align}
where the field strengths $F$ and $F_R$ are {\it anti-hermitian}, $\trtwo$ and $\trfour$ are the traces in the $2$ and $4$ dimensional representations of $Sp(1)_R$. The product factors $2\times$ and $4\times$ are due to the dimensionality of the representations ${\bf 2}$ and ${\bf 4}$. Denoting the representation of the hyperini by $\rho^H$, we use $\tr_H$ as abbreviation for the trace over $\rho^H$. A crucial point to notice is that the overall factors of $-1/2$ in the first three formulae above are due to the fact that the representations involved are negative chiral pseudo-real. In the formula for $P(\psi^+)$ we have not divided by 2 because we are treating the hyperini as Weyl spinors. For a pseudoreal representation \(R\), a half-hypermultiplet contributes one half of the anomaly of a Weyl fermion in \(R\). Next using the trace identities
\begin{align}
\trtwo F_R^4 &= \frac12 \left(\trtwo F_R^2\right)^2\ ,\qquad 
\trfour F_R^2 = 10\,\trtwo F_R^2\ ,\qquad  \trfour F_R^4 = 41\left(\trtwo F_R^2\right)^2 ,
\end{align}
and the well known relation
\be
n_H=n_V+244 ,
\ee
required for the absence of $\tr R^4$ terms in the anomaly polynomial\footnote{This is a necessary condition for the gravitational anomaly to vanish. We impose it here so that our formulae can be directly compared with those in the literature. In the next section, we will systematically list all local and global consistency conditions.}, it is straightforward to find that
\begin{align}
I_8  &= P(\psi_\mu^{A-})+ P(\lambda^{ABC-})+ P(\lambda_i^{A-}) + P(\psi^+)
\nn\w2
&= (\tr R^2)^2 +\frac{1}{12} (n_V-12) \tr R^2 \trtwo F_R^2 + \frac16 \tr R^2 \left( \sum_{i=1}^n \tr_{\rm ad} F_i^2 -\tr_H F^2\right)
\nn\w2
& -\frac16 \big(n_V+84 \big) \left(\trtwo F_R^2 \right)^2 -2 \trtwo F_R^2 \sum_{i=1}^n \tr_{\rm ad} F_i^2 +\frac23 \left( -\sum_{i=1}^n \tr_{\rm ad} F_i^4 +\tr_H F^4\right)\ ,
\label{ap}
\end{align}
in agreement with \cite{Suzuki:2005vu}.
For each ordinary simple factor we define the group-theory coefficients by
\begin{empheq}[box=\fbox]{align}
\tr_R F_i^2 &=A_R\,\tr F_i^2
\nn\\
\tr_R F_i^4 &=B_R\,\tr F_i^4
       +C_R\bigl(\tr F_i^2\bigr)^2 +P_R\,{\cal P}(F_i)
\label{eq:general-trace-coefficients}
\end{empheq}
where \({\cal P}(F_i)\) is the Pfaffian invariant for $D_4$ algebra defined by
\be 
{\cal P}(F) = \frac{1}{2^4 4!} \epsilon^{a_1...a_8} F_{a_1a_2}F_{a_3a_4}F_{a_5a_6}F_{a_7a_8}\ ,
\ee
and should be set to zero when $G_i\neq D_4$. The overall normalization of $P_R$ does not matter for our purpose.

The matter spectrum consists of blocks. A semisimple matter block
is an allowed irreducible product representation
\begin{equation}
 \mathcal R_\mu=(R_{\mu,1},\ldots,R_{\mu,n}),
 \qquad \mathrm{with~dimension}\quad d_\mu=\prod_{i=1}^{n}\dim R_{\mu,i}.
 \label{eq:semisimple-matter-block}
\end{equation}
In terms of the following expressions, 
\begin{align}
 A_i&=A_{{\rm adj},i}-\sum_\mu m_\mu A_{R_{\mu,i}} \prod_{k\ne i}\dim R_{\mu,k}\ , 
 \label{eq:net-A}\w2
 B_i &=B_{{\rm adj},i}-\sum_\mu m_\mu B_{R_{\mu,i}}
       \prod_{k\ne i}\dim R_{\mu,k}\ ,
\label{eq:net-B}\w2
C_i &=C_{{\rm adj},i}-\sum_\mu m_\mu C_{R_{\mu,i}}
       \prod_{k\ne i}\dim R_{\mu,k} \ , 
\label{eq:net-C}\w2
P_i &= P_{{\rm adj},i} -\sum_\mu m_\mu P_{R_{\mu,i}} \prod_{k\neq i}\dim R_{\mu,k}  \ ,
\label{eq:net-P}\w2
M_{ij} &=\sum_\mu m_\mu A_{R_{\mu,i}}A_{R_{\mu,j}}
       \prod_{k\ne i,j}\dim R_{\mu,k}\ ,
\label{a5}
\end{align}
where $m_\mu$ is the multiplicity of the block. In the full-hypermultiplet convention $m_\mu$ is a
nonnegative integer for a real or complex block and a nonnegative
half-integer for a pseudoreal block. The anomaly polynomial is given by 
\begin{align}
 I_8={}& \bigl(\tr R^2\bigr)^2
 +\left(\frac{n_V}{12}-1\right)
   \tr R^2\,\tr_{\mathbf 2}F_R^2 -2\left(\frac{n_V}{12}+7\right)
   \bigl(\tr_{\mathbf 2}F_R^2\bigr)^2 
\nn\\
&  +\sum_i\left\{
 \frac{A_i}{6}\tr R^2\tr F_i^2
 -\frac{2B_i}{3}\tr F_i^4 -\frac23 P_i {\cal P}(F_i)
 -\frac{2C_i}{3}\bigl(\tr F_i^2\bigr)^2\right\}
 \nonumber\\
 &+4\sum_{i<j}M_{ij}\tr F_i^2\tr F_j^2 -2\sum_{i=1}^{n}A_{{\rm adj},i}\,
  \tr F_i^2\,\tr_{\mathbf 2}F_R^2 \ ,
 \label{eq:general-one-loop-I8}
\end{align}
up to a common overall factor suppressed throughout the paper. 

One can further consider gauge groups with a single abelian factor
\be
G= \prod_{i=1}^n G_i \times U(1)\times Sp(1)_R\  .
\label{U1}
\ee
The matter spectrum consists of blocks labeled by $\mu$, carrying U(1) charge $q_\mu$.  Thus,  a matter block is a product representation
\begin{equation}
 \mathcal C_\mu=(R_{\mu,1},\ldots,R_{\mu,n})_{q_\mu}\ ,
 \qquad d_\mu=\prod_{i=1}^{n}\dim R_{\mu,i}\ .
 \label{eq:charged-matter-block}
\end{equation}
In writing down the anomaly polynomial in this case, it is convenient to 
introduce the spectrum-dependent abelian sums
\begin{align}
S_2&=\sum_\mu m_\mu q_\mu^2d_\mu\  ,
&
S_4&=\sum_\mu m_\mu q_\mu^4d_\mu\  ,
\\
S_{i,2}
&=\sum_\mu m_\mu q_\mu^2 A_{R_{\mu,i}}
  \prod_{k\neq i}\dim R_{\mu,k} \ ,
&
S_{i,1}
&=\sum_\mu m_\mu q_\mu E_{R_{\mu,i}}
  \prod_{k\neq i}\dim R_{\mu,k}\  .
\end{align}
where $E_R$ is defined by the following trace relation
\be
\tr_R F_i^3=E_R\,\tr F_i^3 \ ,
\ee
and is nontrivial only when \(G_i\) possesses a cubic invariant. The anomaly polynomial in this case is given by
\begin{align}
I_8={}&
 \bigl(\tr R^2\bigr)^2
 +\left(\frac{n_V}{12}-1\right)
  \tr R^2\,\tr_{\mathbf 2}F_R^2
 -2\left(\frac{n_V}{12}+7\right)
  \bigl(\tr_{\mathbf 2}F_R^2\bigr)^2
\nonumber\\
&+\sum_{i=1}^{n}
 \left[
  \frac{A_i}{6}\tr R^2\,\tr F_i^2
 -\frac{2B_i}{3}\tr F_i^4 -\frac23 P_i {\cal P}(F_i)
 -\frac{2C_i}{3}\bigl(\tr F_i^2\bigr)^2
 \right]
\nonumber\\
&+4\sum_{i<j}M_{ij}\,
  \tr F_i^2\,\tr F_j^2
-2\sum_{i=1}^{n}A_{{\rm adj},i}\,
  \tr F_i^2\,\tr_{\mathbf 2}F_R^2 
  \nonumber\\
  & -\frac{S_2}{6}\tr R^2\,F^2
 +\frac{2S_4}{3}F^4
+4\sum_{i=1}^{n}S_{i,2}\,
  F^2\,\tr F_i^2 +\frac{8}{3}\sum_{i=1}^{n}S_{i,1} F\,\tr F_i^3\ ,
\label{eq:complete-I8-U1-Sp1R}
\end{align}
where  \(F_i\), \(F\), and \(F_R\) denote the field strengths of \(G_i\),
the ordinary \(U(1)\), and \(Sp(1)_R\), respectively. There is no term proportional to \(F^2\tr_{\mathbf 2}F_R^2\), because no chiral fermion is charged
simultaneously under the ordinary \(U(1)\) and \(Sp(1)_R\).

\section{Consistency conditions}
\label{sec: consistency conditions}

We call a model consistent if (1) its local anomaly polynomial admits Green--Schwarz factorization, (2) its anomaly coefficients satisfy the quantization conditions required for a globally well-defined Green--Schwarz term, and (3) its gauge kinetic terms are positive in a nonempty domain of tensor-moduli space. If the gauge group consists of direct
products of $U(1)$ factors and simply connected simple compact groups, which shall be the case in this paper\footnote{For such groups, the computation in Appendix~A of \cite{Becker:2025xgy} gives \(\Omega^{\mathrm{Spin}}_7(BG)=0\), hence any residual bordism anomaly must also vanish.}, the conditions (1) and (2) amount to the following~\cite{Monnier:2017oqd,Monnier:2018nfs}:
(i) the anomaly polynomial must factorize as \eqref{gs}; (ii) the gravitational anomaly coefficient must be a characteristic element of the string-charge
lattice,
\begin{equation}
 a\cdot x\equiv x\cdot x\pmod 2
 \qquad \text{for every }x\in\Lambda ,
\label{eq:characteristic-element}
\end{equation}
where $a$ is defined in \eqref{eq: factorization} and $\Lambda$ is the string charge lattice; (iii) all anomaly vectors must be embedded in a unimodular rank-two lattice of signature $(1,1)$, as shall be detailed below\footnote{When there is an abelian factor, $b_{U(1)}/2\in\Lambda$ must also be satisfied, where $b_{U(1)}$ is defined in \eqref{eq: factorization}.}. The condition (ii) will be manifestly satisfied by all models we find in this paper. The condition (iii) turns out to put a restriction on the dimension $n_V$ of the gauge group, as we shall see below. Finally, the condition (3) on the positivity of the full gauge-kinetic matrix in a nonempty domain will exclude neutral \(U(1)\) factors, or drones.

\subsection{Local anomaly-free condition}

As is well known, the Green-Schwarz mechanism requires that all irreducible fourth-order invariants must vanish. In addition, if $G_{{\rm non}-R}$ contains a group that has a cubic invariant, then terms of the form $F_{U(1)} \tr F_i^3$ may arise, in which case their sum must also vanish. These conditions can be summarized as
\be
n_H=n_V+244\ ,\quad B_i=0\ , \qquad P_i=0\ ,\quad S_{i,1}=0\ .
\label{ga}
\ee
When the irreducible fourth-order anomaly terms vanish, one can write 
\be
I_8 = \sum_{x,y} M_{xy} K^x K^y \ ,\qquad K^x := (\tr R^2, \trtwo F_R^2, \tr F_i^2, F^2)\ ,
\label{fa}
\ee
where the traces without subscripts are in the fundamental representation of $G_i$, and $F$ is the $U(1)$ field strength. Now factorization of $I_8$ requires rank$(M) \le 2$ and that product of its eigenvalues $\lambda_1\lambda_2 \le 0$. When these conditions are satisfied we can write
\begin{equation}
  I_8={1\over2}\eta_{\alpha\beta} Y^\alpha Y^\beta\ ,
\label{gs}
\end{equation}
with
\be
Y^\alpha
 =\frac12a^\alpha\tr R^2
 +\sum_{i=1}^{n}\frac{2b_i^\alpha}{\lambda_i}\tr F_i^2
 +2b_{U(1)}^\alpha F^2
 +2b_R^\alpha\tr_{\mathbf 2}F_R^2 .
 \label{eq: factorization}
\ee
up to an overall factor which we suppress, and the $\lambda$-coefficients are given in \autoref{tab:lambda}. We set $\lambda=1$ for $U(1)$. 

\begin{table}[h]
    \centering
\begin{center}
    \begin{tabular}{|c|c|c|c|c|c|c|c|c|c|}
\hline
 & $A_n$ & $B_n$ & $C_n$ & $D_n$ & $E_6$ & $E_7$ & $E_8$ & $F_4$ & $G_2$\\
\hline
$\lambda$ & 1 & 2 & 1 & 2 & 6 & 12 & 60 & 6 & 2\\
\hline
    \end{tabular}
    \caption{\small Normalization factor $\lambda$ for different groups.}
    \label{tab:lambda}
\end{center}
\end{table}
Recalling \eq{ga}, and matching \eqref{eq:general-one-loop-I8} (or \eqref{eq:complete-I8-U1-Sp1R}) to \eqref{gs}, altogether we get the anomaly equations
\begin{subequations}
 \begin{align}
 &n_H =n_V+244\ ,
 \label{ab}\w2
&B_i =0\ , \qquad P_i=0\ ,\quad S_{i,1}=0\ ,
\label{c0}\w2
&a^2 =8\ ,\qquad p_R :=a\cdot b_R=\frac{n_V}{12}-1\ , \qquad b_R^2 =-\left(\frac{n_V}{12}+7\right)\ ,
\label{c1} \w2
&p_i  := a\cdot b_i =\frac{\lambda_i A_i}{6} \quad \implies \quad A_i= \frac{6p_i}{\lambda_i}\ ,
\label{pi}\w2
&q_i  := b_i\cdot b_R=-\frac{\lambda_i A_{{\rm adj},i}}{2}\ ,
\label{qi} \w2
&b_i^2 =-\frac{\lambda_i^2 C_i}{3} \quad \implies \quad C_i=-\frac{3b_i^2}{\lambda_i^2}\ ,
\w2
&b_i\cdot b_j =\lambda_i\lambda_j M_{ij}\ ,
\label{eq:general-local-anomaly-equations}\w2
&a\cdot b_{U(1)}=-\frac{S_2}{6}\ , \qquad b_{U(1)}^2=\frac{S_4}{3}\ ,\qquad  b_i\cdot b_{U(1)}=\lambda_i S_{i,2}\ ,
\label{BL}\w2
&b_{U(1)}\cdot b_R =0\ .
\label{ac1}
\end{align}
\end{subequations}
Among these equations those which depend on the multiplicities, denoted by $m_\mu$, are the following 
\begin{subequations}
\label{eq:anomaly-targets}
\begin{empheq}[box=\fbox]{align}
&\widehat A_i = A_{{\rm adj},i}-\frac{6p_i}{\lambda_i}, \qquad \widehat B_i = B_{{\rm adj},i},
\qquad
\widehat C_i = C_{{\rm adj},i}+\frac{3b_i^2}{\lambda_i^2}, \qquad \widehat P_i = P_{{\rm adj},i},
\label{ABC}\\
&M_{ij} = \frac{b_i\cdot b_j}{\lambda_i\lambda_j},
\label{link}\\
&a\cdot b_{U(1)} = -\frac{S_2}{6},
\qquad
b_{U(1)}^2=\frac{S_4}{3}, \qquad b_i\cdot b_{U(1)}=\lambda_i S_{i,2}\ ,\qquad S_{i,1}=0\ ,
\label{eq:U1-anomaly}
\end{empheq}
\end{subequations}
where  
\begin{subequations}
\begin{align}
{\widehat A}_i &:= \sum_\mu m_\mu A_{R_{\mu,i}} \prod_{k\ne i}\dim R_{\mu,k}\ , 
 \label{h1}\w2
{\widehat B}_i &:=\sum_\mu m_\mu B_{R_{\mu,i}} \prod_{k\ne i}\dim R_{\mu,k}\ ,
\label{h2}\w2
{\widehat C}_i &:=\sum_\mu m_\mu C_{R_{\mu,i}}
       \prod_{k\ne i}\dim R_{\mu,k} \ , 
\label{h3}\w2
{\widehat P}_i &:= \sum_\mu m_\mu P_{R_{\mu,i}} \prod_{k\neq i}\dim R_{\mu,k}\ ,
\label{h4}\w2
M_{ij} &=\sum_\mu m_\mu A_{R_{\mu,i}}A_{R_{\mu,j}}
       \prod_{k\ne i,j}\dim R_{\mu,k}\ .
\label{h5}
\end{align} 
\end{subequations}
Among the remaining anomaly equations, \eq{c1} together with the requirement of global anomaly freedom, will be shown to restrict $n_V=12$ or $n_V=96$.  \eq{qi} and \eq{ac1} put constraints on the anomaly-vector classes, and more specifically, they relate the two components of $b$-vector. Once an anomaly-vector class is fixed (see Sections \ref{sec: a search strategy} and \ref{sec: nV=96 branch} for details), finding the anomaly-free spectra reduces to solving \eqref{eq:anomaly-targets} for the matter multiplicities alone. 

\subsection{Unimodularity and the condition on $n_V$}

Importantly, even if the local anomaly cancellations are satisfied we must demand the embedding of all anomaly vectors in a unimodular rank-two lattice of signature $(1,1)$, by the requirement of freedom from certain global anomalies, extensively discussed in \cite{Monnier:2017oqd}, and especially in the context of $U(1)_R$ gauged $6D$ supergravities in \cite{Becker:2023zyb}, as mentioned earlier. In the case of $Sp(1)_R$ gauging the unimodularity requirement imposes a condition on the number $n_V$ of the vector multiplets \cite{Becker:2025xgy} as can be seen as follows. From \eq{c1}, we have
\begin{align}
a\cdot a &= 8\ ,
\w2
a\cdot b_R &= \frac{1}{12}(n_{V} - 12) := p_R \ ,
\w2
b_R\cdot b_R &= -\frac{1}{12}(n_{V} + 84)=-(p_R+8)\ .
\end{align}

Since \(a\) and \(b_R\) are vectors in the unimodular charge lattice
\(\Lambda\), the Gram determinant of the sublattice they span must be
the determinant of \(\Lambda\), namely \(-1\), times a perfect square \cite{Seiberg:2011dr}.
Equivalently, defining 
\be
K = \left( \begin{array}{cc}
a\cdot a & a\cdot b_R \\
a\cdot b_R & b_R\cdot b_R
\end{array} \right)\ ,
\label{defK}
\ee
we have
\be
\det K = \det \left( \begin{array}{cc}
8 & p_R \\
p_R & -(p_R+8)
\end{array} \right)=  -m^2\ ,\qquad m, p_R\in \mathbb{Z}\ .
\ee
This gives $(m-p_R-4)(m+p_R+4)=48$, and noting that $n_V=12(p_R+1)>0$, it has the solution $p_R=7,0$. Thus the unimodularity of the rank-two charge lattice already restricts the \(Sp(1)_R\)-gauged models to the vector-multiplet numbers \cite{Becker:2025xgy} $n_V=12$ or $n_V=96$.

\subsection{Ghost-free condition}

A pseudo-action for $6D$, $\mathcal{N}=(1,0)$, $Sp(1)_R$ gauged supergravity coupled to a tensor, vector and hypermultiplet can be found in \cite{Riccioni:2001bg}. The terminology of ``pseudo'' refers to the fact that all but the (anti)self-duality conditions on the 3-form field strengths follow from the action. A proper action, albeit non-manifestly Lorentz covariant, has been constructed in \cite{Bossard:2024ffp}. For our purposes here, it suffices to consider the Yang-Mills kinetic terms for the $Sp(1)_R$ gauged models we are studying here. In our conventions, these kinetic terms take the form 
\be
e^{-1} {\cal L}_{YM} = -\frac14 \left( v_z^+ e^\phi +v_z^- e^{-\phi}\right) F_{z\mu\nu}^r  F_z^{r \mu\nu} \ ,
\ee
where $r$ labels the adjoint representation,  $z=(i,R)$ labels the group factors and the vectors $v_z^\alpha = (v_z^+, v_z^-)$ are proportional to the anomaly vectors as discussed in detail in \cite{Guo:2025mlb}. The proportionality constant can be taken to be positive, and as such the ghost-free gauge kinetic terms impose the conditions
\be
\mbox{ghost freedom}\ \implies\  b_z^+ e^\phi +b_z^- e^{-\phi} >0\ ,\quad \forall z\ .
\label{gf}
\ee
Thus, if \(b_z^+\) and \(b_z^-\) are both nonnegative and at least one is strictly positive, the corresponding gauge kinetic term is positive for every finite value of the dilaton. If they have opposite signs, positivity restricts the dilaton to an interval, and ghost freedom requires a nonempty common interval for all gauge factors.

\subsection{Pseudo-action and the no-drone condition}

We now explain why \(U(1)\) factors under which all fields are neutral cannot furnish ordinary propagating vector multiplets. 
Although our scans contain  at most one non-$R$ abelian factor, we shall explain the drone-free condition in a more general setting. Use index $I=1,...,M$ to label the abelian vector fields. 
The relevant part of the  pseudo action constructed in \cite{Riccioni:2001bg} that describes the coupling of Yang-Mills multiplets and a number of abelian  vector multiplets is 
\be
{\cal L}_D= j\cdot b_z\,\tr_z \left(F\wedge ^\star F\right) -\frac12 j\cdot b_{IJ} F^I \wedge ^\star F^J
+B\cdot b_z \wedge \tr_z \left( F\wedge F\right)-\frac12 B\cdot b_{IJ} \,F^I \wedge F^J\ .
\ee
where the anomaly vectors $b_{IJ}^\alpha = b_{JI}^\alpha$ and $j\cdot b_z= \eta_{\alpha\beta}\,j^\alpha b_z^\beta$. 
If $j^\alpha$ is the unit positive tensor scalar,  the abelian gauge kinetic
matrix which is proportional to $j\cdot b_{IJ}$ must be positive definite for all $M$ abelian vectors to propagate
unitarily.

In the case of $M$ abelian vector fields, it suffices to focus on the generalization of the first two anomaly equations in \eq{BL}, which now take the form
\begin{align}
a\cdot b_{IJ}
&=-\frac16\sum_\mu m_\mu d_\mu q_{\mu I}q_{\mu J},
\label{eq:multi-u1-grav}\\
b_{IJ}\cdot b_{KL}+b_{IK}\cdot b_{JL}+b_{IL}\cdot b_{JK}
&=\sum_\mu m_\mu d_\mu q_{\mu I}q_{\mu J}q_{\mu K}q_{\mu L}.
\label{eq:multi-u1-quartic}
\end{align}
Choose a basis in which $U(1)_D$ is a drone, so $q_{\mu D}=0$ for every $\mu$.  Equation
\eqref{eq:multi-u1-grav} for $(DD)$ and $(DI)$ indices, and equation \eqref{eq:multi-u1-quartic}, with indices $(DDDD)$ and $(DDII)$ give
\begin{align}
 a\cdot b_{DD} &=0\ ,
 \label{eq:multi-u1-DD}\\
 a\cdot b_{DI}&=0\ ,
 \label{de1}\\
 b_{DD}^2 &=0\ ,
 \label{de2}\\
 b_{DD}\cdot b_{II}+2b_{DI}^2&=0\ .
 \label{eq:multi-u1-DDII}
\end{align}
Since $a^2=8>0$, $a$ is spacelike.  Indeed, choosing $a=(-2,-2)$, it readily follows from  \eq{eq:multi-u1-DD} and \eqref{de2} that $b_{DD}=0$.  It then follows
from \eqref{de1} and \eq{eq:multi-u1-DDII} that 
\begin{equation}
 b_{DI}=0
 \qquad\hbox{for every }I.
 \label{eq:multi-u1-zero-row}
\end{equation}
Thus, the $D$-th row and column of $j\cdot b_{IJ}$ vanish, and consequently the
kinetic matrix is degenerate, even when the most general off-diagonal
$b_{IJ}$ are initially allowed.  

If the number of tensor multiplets $n_T\ge 9$, then $a^2= 9-n_T \le 0$. In this case the above argument fails
because $a$ is null or timelike. Park and Taylor \cite{Park:2011wv} show that a nonzero null vector $b$ may then satisfy
$a\cdot b=0$ and $j\cdot b>0$; taking $b_{IJ}=\delta_{IJ}b$ can give a
positive-definite kinetic matrix for arbitrarily many neutral $U(1)$ factors.

\section{A search strategy}
\label{sec: a search strategy}

We now formulate a search strategy that applies both to the
semisimple gauge groups in \eqref{g1} and, after a finite charge decoration of
the matter representations, to \eqref{U1}. For the former, we supplement this strategy with the spectra-assembly method to substantially reduce the computational workload. The details will be discussed in Section \ref{sec: nV=96 branch}. 
We first classify the anomaly vectors in the rank-two string-charge lattice. For each simple group factor, the anomaly equations and lattice integrality constrain the allowed gauge anomaly vectors.  Each surviving choice
then fixes the RHS of the pure and mixed gauge-anomaly equations given in \eqref{eq:anomaly-targets}. When there is a charged abelian factor, constraints on the corresponding anomaly vector can be derived similarly and the details are given in Section \ref{sec:a2u1}.

We next ask whether those candidate anomaly vectors can be realized by a physical
hypermultiplet spectrum.  Every allowed product representation contributes a
definite column of group-theory coefficients.  Its physical multiplicity is
denoted by $m_\mu$ and is measured in {\it full-hypermultiplet} units. These multiplicities must solve the
anomaly equations subject to the hypermultiplet reality conditions.  

To be more precise, for a fixed gauge algebra, representation domain, and anomaly-vector
class \(\Sigma\), let \(m_\mu\) denote the physical multiplicity of the
product representation \(\mathcal R_\mu\), measured in
full-hypermultiplet units.  Thus we have 
\begin{equation}
 m_\mu\in\mathbb Z_{\ge0}
 \quad\hbox{for real or complex }\mathcal R_\mu,
 \qquad
 m_\mu\in\frac12\mathbb Z_{\ge0}
 \quad\hbox{for pseudoreal }\mathcal R_\mu.
 \label{eq:physical-multiplicity-rule}
\end{equation}

\subsection*{The anomaly equations in matrix form}

The anomaly equations are linear in the multiplicities as unknown variables, and the
spectrum satisfies the gravitational anomaly equation $n_H=244+n_V$. Schematically, these equations can be written as
\begin{equation}
\begin{aligned}
    Dm&=t_\Sigma,\\
    d^{\mathsf T}m+s&=n_H, \qquad s\in\mathbb Z_{\ge0},
\end{aligned}
 \label{ga-physical}
\end{equation}
where \(s\) is the number of hypermultiplets neutral under the non-\(R\)
gauge group, the matrix $D$ collects the anomaly data of the allowed matter blocks: each column corresponds to a product representation, with entries given by the relevant group-theory coefficients, $t_\Sigma$ organizes the RHS of \eqref{eq:anomaly-targets} for a given anomaly-vector class into a column, and the vector $d$ encodes the dimension of each product irreducible representation.

More explicitly, define the nonabelian entries contributed by one matter block by
\begin{align}
\mathcal A_{i\mu}
&:=A_{R_{\mu i}}\prod_{k\neq i}\dim R_{\mu k},
&
\mathcal B_{i\mu}
&:=B_{R_{\mu i}}\prod_{k\neq i}\dim R_{\mu k},\\
\mathcal C_{i\mu}
&:=C_{R_{\mu i}}\prod_{k\neq i}\dim R_{\mu k},
&
\mathcal P_{i\mu}
&:=P_{R_{\mu i}}\prod_{k\neq i}\dim R_{\mu k},\\
\mathcal M_{ij,\mu}
&:=A_{R_{\mu i}}A_{R_{\mu j}}
\prod_{k\neq i,j}\dim R_{\mu k}.
\end{align}
Let $\{i_1,\ldots,i_r\}$ label the set of simple factors possessing a cubic invariant. Then the complete matrix and target vector can be written as
\begin{equation}
\resizebox{0.75\linewidth}{!}{$\displaystyle
\begin{aligned}
D &=
\begin{pmatrix}
\mathcal A_{11} & \mathcal A_{12} & \cdots & \mathcal A_{1N}\\
\mathcal B_{11} & \mathcal B_{12} & \cdots & \mathcal B_{1N}\\
\mathcal C_{11} & \mathcal C_{12} & \cdots & \mathcal C_{1N}\\
\mathcal P_{11} & \mathcal P_{12} & \cdots & \mathcal P_{1N}\\
\vdots & \vdots & & \vdots\\
\mathcal A_{n1} & \mathcal A_{n2} & \cdots & \mathcal A_{nN}\\
\mathcal B_{n1} & \mathcal B_{n2} & \cdots & \mathcal B_{nN}\\
\mathcal C_{n1} & \mathcal C_{n2} & \cdots & \mathcal C_{nN}\\
\mathcal P_{n1} & \mathcal P_{n2} & \cdots & \mathcal P_{nN}\\[1mm]
\mathcal M_{12,1} & \mathcal M_{12,2} & \cdots & \mathcal M_{12,N}\\
\mathcal M_{13,1} & \mathcal M_{13,2} & \cdots & \mathcal M_{13,N}\\
\vdots & \vdots & & \vdots\\
\mathcal M_{n-1,n,1} & \mathcal M_{n-1,n,2}
 & \cdots & \mathcal M_{n-1,n,N}\\[1mm]
\mathcal E_{i_1 1} & \mathcal E_{i_1 2} & \cdots & \mathcal E_{i_1 N}\\
\vdots & \vdots & & \vdots\\
\mathcal E_{i_r 1} & \mathcal E_{i_r 2} & \cdots & \mathcal E_{i_r N}\\[1mm]
\mathcal U_{11} & \mathcal U_{12} & \cdots & \mathcal U_{1N}\\
\vdots & \vdots & & \vdots\\
\mathcal U_{n1} & \mathcal U_{n2} & \cdots & \mathcal U_{nN}\\[1mm]
\mathcal Q_{2,1} & \mathcal Q_{2,2} & \cdots & \mathcal Q_{2,N}\\
\mathcal Q_{4,1} & \mathcal Q_{4,2} & \cdots & \mathcal Q_{4,N}
\end{pmatrix},
\qquad
t_\Sigma &=
\begin{pmatrix}
\widehat A_1\\
\widehat B_1\\
\widehat C_1\\
\widehat P_1\\
\vdots\\
\widehat A_n\\
\widehat B_n\\
\widehat C_n\\
\widehat P_n\\[1mm]
M_{12}\\
M_{13}\\
\vdots\\
M_{n-1,n}\\[1mm]
0\\
\vdots\\
0\\[1mm]
\dfrac{b_1\cdot b_U}{\lambda_1}\\
\vdots\\
\dfrac{b_n\cdot b_U}{\lambda_n}\\[3mm]
-6\,a\cdot b_U\\[1mm]
3\,b_U^2
\end{pmatrix}.
\end{aligned}
$}
\label{eq:master-matrix-system}
\end{equation}
Here $n$ and $N$ are the numbers of simple factors and matter blocks, respectively. A $\mathcal P_i$ row and the corresponding $\widehat P_i$ target are present only for $G_i=D_4$. For a purely semisimple gauge group, the rows beginning with $\mathcal E_{i_a\mu}$ and all subsequent $U(1)$-dependent rows are absent.

When a single charged $U(1)$ factor is present, the following charge-dependent rows are appended:
\begin{align}
\mathcal E_{i\mu}
&:=q_\mu E_{R_{\mu i}}\prod_{k\neq i}\dim R_{\mu k},
\quad\text{for $G_i$ with a cubic invariant},
\\
\mathcal U_{i\mu}
&:=q_\mu^2 A_{R_{\mu i}}\prod_{k\neq i}\dim R_{\mu k},
\\
\mathcal Q_{2,\mu}
&:=q_\mu^2d_\mu,\\
\mathcal Q_{4,\mu}
&:=q_\mu^4d_\mu.
\end{align}
The first row encodes the irreducible $U(1)-G_i^3$ anomaly, the second the
$U(1)^2-G_i^2$ anomaly, and the last two the $\mathrm{gravitational}-U(1)^2$ and
$U(1)^4$ anomalies.

\section{The \(n_V=96\) semisimple branch}
\label{sec: nV=96 branch}

In this branch, the group $G_{ss}=\prod_{i=1}^nG_i \times Sp(1)_R$ is semisimple, it has dimension $n_V=96$, and $A_1,A_2$ are excluded. The representation domain consists of all irreducible representations of the group $\prod_{i=1}^nG_i$ whose dimensions are compatible with $n_H = 244+n_V=340$\footnote{Pseudoreal irreducible representations with dimensions between 341 and 680 are also included since they allow half-integer multiplicities.}.

\subsection{Conditions on anomaly vectors from unimodular lattice embedding}

We use light-cone coordinates on the rational anomaly space, in which
\be
 \eta=
 \begin{pmatrix}0&1\\1&0\end{pmatrix}\ ,
 \qquad
 a=(-2,-2),
 \qquad
 b_R=\left(\frac32,-5\right)\ .
\ee
Thus, for \(u=(u_1,u_2)\) and \(v=(v_1,v_2)\)\ ,
\be
 u\cdot v=u_1v_2+u_2v_1\ .
\ee
The coordinate vectors \((1,0)\) and \((0,1)\) are not being taken as
an integral basis of the string-charge lattice.  Since
\(b_R^2=-15\) is odd, the unimodular lattice is necessarily the
odd one \(I_{1,1}\).

An explicit integral basis is
\be
 e=\frac{4a+b_R}{13}
   =\left(-\frac12,-1\right)\ ,
 \qquad
 f=a-3e
   =\left(-\frac12,1\right)\ ,
\ee
for which
\be
 e^2=1,\qquad f^2=-1,\qquad e\cdot f=0\ .
\ee
Hence
\be
\Lambda=\mathbb Ze\oplus\mathbb Zf\simeq I_{1,1}\ , \qquad a=3e+f,\qquad b_R=e-4f\ .
\ee
Writing
\be
b_i=x_i e+y_i f\ ,
\ee
and defining
\be
\boxed{ p_i=a\cdot b_i}\ , \qquad \boxed{q_i=b_R\cdot b_i}\ ,
\ee
we have, since \(a=3e+f\) and \(b_R=e-4f\),
\be
p_i=3x_i-y_i\ , \qquad q_i=x_i+4y_i\ .
\ee
Solving for the integral-basis coordinates gives
\be
 x_i=\frac{4p_i+q_i}{13}\ , \qquad y_i=\frac{3q_i-p_i}{13}\ .
\ee
The lattice condition  \( b_i\in\Lambda \) is therefore equivalent to
\be
 x_i,y_i\in\mathbb Z\ ,
\ee
or, equivalently,
\be 
\boxed{p_i\equiv3q_i\pmod{13}}\ .
\ee

Since \(e^2=1\), \(f^2=-1\), and \(e\cdot f=0\), from $b_i\cdot b_j
 =x_ix_j-y_iy_j $ one finds 
\begin{align}
 b_i\cdot b_j
 &=
 \frac{
 15p_ip_j+7(p_iq_j+q_ip_j)-8q_iq_j
 }{169}\ .
 \label{eq:bi-bj}
\end{align}
and $b_i^2 = b_i \cdot b_j \vert_{j=i}$.

\subsection{Semisimple groups of dimension $93$}

Before imposing any local anomaly-vector test, the simple factor types of
dimension at most \(93\), with \(A_1\), \(A_2\), and \(U(1)\) omitted, are
\be
 \mathcal S_{93}=\{A_3,\ldots,A_8;\ B_3,\ldots,B_6;
 C_2,\ldots,C_6; \ D_4,\ldots,D_7;\ G_2,F_4,E_6\}\ .
\ee
Enumerating unordered products of factors from \(\mathcal S_{93}\) whose dimensions sum to \(93\) gives 79 raw semisimple gauge-algebra sectors.

\subsection{Origin of the one-factor targets}

A complete matter block is a product representation
\be
 \mathcal R_\mu=\bigotimes_{j=1}^n R_{j\mu}\ ,
\ee
with physical multiplicity \(m_\mu\) in full-hypermultiplet units.  Its
contribution to the pure anomaly equations of a fixed \(G_i\) is described by
\be
 A_{i\mu} =A_{R_{i\mu}}\prod_{j\ne i}\dim R_{j\mu},\qquad
 B_{i\mu} =B_{R_{i\mu}}\prod_{j\ne i}\dim R_{j\mu},\qquad
 C_{i\mu} =C_{R_{i\mu}}\prod_{j\ne i}\dim R_{j\mu}\ .
\ee
Consequently, the one-factor anomaly equations inherited from a multi-factor spectrum are
\be
 \sum_\mu m_\mu A_{i\mu}=\widehat A_i,\qquad
 \sum_\mu m_\mu B_{i\mu}=\widehat B_i,\qquad
 \sum_\mu m_\mu C_{i\mu}=\widehat C_i\ .
 \label{eqn: one-factor anomaly equations}
\ee

For the relaxed one-factor problem, it is convenient to collect all
copies of a given irreducible representation \(R\) of \(G_i\) into the
effective multiplicity
\be
 m_R^{(i)}
 := \sum_{\mu:\,R_{i\mu}=R} m_\mu\prod_{j\ne i}\dim R_{j\mu}\ .
 \label{eq: effective multiplicity projection}
\ee
The same equations may then be written
\be
 \sum_R m_R^{(i)}A_R=\widehat A_i\ ,\qquad
 \sum_R m_R^{(i)}B_R=\widehat B_i\ ,\qquad
 \sum_R m_R^{(i)}C_R=\widehat C_i\ .
\ee
with
\be
 \widehat A_i=A_{\mathrm{adj},i}-\frac{6p_i}{\lambda_i}\ ,
 \qquad
 \widehat B_i=B_{\mathrm{adj},i}\ ,
 \qquad
 \widehat C_i=C_{\mathrm{adj},i}
       +\frac{3b_i^2}{\lambda_i^2}\ .
\ee
For \(D_4\), the Pfaffian equation is imposed as a separate signed row.
These hatted quantities therefore determine the matter side of the anomaly equations, namely LHS of \eqref{ABC}.

For fixed \(G\), the mixed \(G^2-Sp(1)_R\) equation first fixes \(q=-\frac{\lambda}{2}A_{\rm adj}\).  Lattice
integrality then restricts \(p\) to be of the form
\be
 p=3q+13\ell,\qquad \ell\in\mathbb Z\ .
 \label{LI}
\ee
Although this congruence defines an infinite arithmetic progression, the $n_H$ bound makes the relevant search finite.  Indeed, if
\be
 \kappa_G:= \max_{R\in\mathcal R_G}\frac{C_R}{\dim R}\ ,
\ee
where \(\mathcal R_G\) is the finite representation domain used in the
scan, then every one-factor solution must satisfy
\be
 \widehat C_G(p)
 =\sum_R m_RC_R \leq \kappa_G\sum_Rm_R\dim R \leq340\kappa_G\ .
 \label{eqn: quadratic in p}
\ee
Since \(\widehat C_G(p)\) is quadratic in \(p\), this inequality leaves
only finitely many values in the required congruence class. Note that \eqref{eqn: quadratic in p} provides only necessary conditions, and not every candidate admits a physical spectrum.

\subsection{The 48 one-factor candidate anomaly-vector classes}

The factors possessing at least one anomaly-vector class are 
\be
 \{A_3,\ldots,A_7;\ B_3,B_4,B_5;\ C_2,\ldots,C_6;\
   D_4,D_5,D_7;\ G_2,E_6\}\ .
\ee
The remaining raw factor types \(A_8,B_6,D_6,F_4\) possess no admissible $b$-vectors in the stated representation and dimension domain.  (The types \(A_8\) and \(D_6\) in fact occur in no dimension-93 product.)
Table~\ref{tab:local48} gives every retained class.
The last column gives the coordinates \((x,y)\) in \(b=xe+yf\).

\begingroup
\scriptsize
\setlength{\tabcolsep}{2.7pt}
\renewcommand{\arraystretch}{1.08}
\begin{center}
\begin{tabular}{r c r r c @{\hspace{0.42in}} r c r r c}
\toprule
\#&\(G\)&\(p\)&\(q\)&\((x,y)\)&
\#&\(G\)&\(p\)&\(q\)&\((x,y)\)\\
\midrule
1 & $A_3$ & $-77$ & $-4$ & $(-24,5)$ & 25 & $C_2$ & $-87$ & $-3$ & $(-27,6)$ \\
2 & $A_3$ & $-64$ & $-4$ & $(-20,4)$ & 26 & $C_2$ & $-74$ & $-3$ & $(-23,5)$ \\
3 & $A_3$ & $-51$ & $-4$ & $(-16,3)$ & 27 & $C_2$ & $-61$ & $-3$ & $(-19,4)$ \\
4 & $A_3$ & $-38$ & $-4$ & $(-12,2)$ & 28 & $C_2$ & $-48$ & $-3$ & $(-15,3)$ \\
5 & $A_3$ & $-25$ & $-4$ & $(-8,1)$ & 29 & $C_2$ & $-35$ & $-3$ & $(-11,2)$ \\
6 & $A_3$ & $-12$ & $-4$ & $(-4,0)$ & 30 & $C_2$ & $-22$ & $-3$ & $(-7,1)$ \\
7 & $A_4$ & $-28$ & $-5$ & $(-9,1)$ & 31 & $C_2$ & $-9$ & $-3$ & $(-3,0)$ \\
8 & $A_4$ & $-15$ & $-5$ & $(-5,0)$ & 32 & $C_3$ & $-25$ & $-4$ & $(-8,1)$ \\
9 & $A_4$ & $-2$ & $-5$ & $(-1,-1)$ & 33 & $C_3$ & $-12$ & $-4$ & $(-4,0)$ \\
10 & $A_5$ & $-18$ & $-6$ & $(-6,0)$ & 34 & $C_4$ & $-15$ & $-5$ & $(-5,0)$ \\
11 & $A_5$ & $-5$ & $-6$ & $(-2,-1)$ & 35 & $C_4$ & $-2$ & $-5$ & $(-1,-1)$ \\
12 & $A_6$ & $-21$ & $-7$ & $(-7,0)$ & 36 & $C_5$ & $-5$ & $-6$ & $(-2,-1)$ \\
13 & $A_6$ & $-8$ & $-7$ & $(-3,-1)$ & 37 & $C_6$ & $-8$ & $-7$ & $(-3,-1)$ \\
14 & $A_7$ & $-24$ & $-8$ & $(-8,0)$ & 38 & $D_4$ & $-18$ & $-6$ & $(-6,0)$ \\
15 & $A_7$ & $-11$ & $-8$ & $(-4,-1)$ & 39 & $D_4$ & $-5$ & $-6$ & $(-2,-1)$ \\
16 & $B_3$ & $-41$ & $-5$ & $(-13,2)$ & 40 & $D_5$ & $-24$ & $-8$ & $(-8,0)$ \\
17 & $B_3$ & $-28$ & $-5$ & $(-9,1)$ & 41 & $D_5$ & $-11$ & $-8$ & $(-4,-1)$ \\
18 & $B_3$ & $-15$ & $-5$ & $(-5,0)$ & 42 & $D_5$ & $2$ & $-8$ & $(0,-2)$ \\
19 & $B_3$ & $-2$ & $-5$ & $(-1,-1)$ & 43 & $D_7$ & $-10$ & $-12$ & $(-4,-2)$ \\
20 & $B_4$ & $-21$ & $-7$ & $(-7,0)$ & 44 & $G_2$ & $-64$ & $-4$ & $(-20,4)$ \\
21 & $B_4$ & $-8$ & $-7$ & $(-3,-1)$ & 45 & $G_2$ & $-38$ & $-4$ & $(-12,2)$ \\
22 & $B_5$ & $-14$ & $-9$ & $(-5,-1)$ & 46 & $G_2$ & $-25$ & $-4$ & $(-8,1)$ \\
23 & $C_2$ & $-113$ & $-3$ & $(-35,8)$ & 47 & $G_2$ & $-12$ & $-4$ & $(-4,0)$ \\
24 & $C_2$ & $-100$ & $-3$ & $(-31,7)$ & 48 & $E_6$ & $-10$ & $-12$ & $(-4,-2)$ \\
\bottomrule
\end{tabular}
\captionof{table}{The 48 one-factor candidate anomaly-vector classes in the
\(n_V=96\) semisimple branch.  Each row entry represents
\((G,p,q)\), with \(b=xe+yf\). Note that these are candidate anomaly-vector classes and do not necessarily give rise to anomaly-free spectra.}
\label{tab:local48}
\end{center}
\endgroup

\subsection{Gauge sectors and global candidate classes}

    For a gauge sector containing \(m_{G_i}\) copies of a factor \(G_i\), let \(s_{G_i}\)
be the number of its entries in Table~\ref{tab:local48}.  Permuting identical
factors does not change the class, so the number of total anomaly-vector classes for $G_{ss}$ is
\be
 N(\mathfrak g)
 =\prod_i\binom{s_{G_i}+m_{G_i}-1}{m_{G_i}}\ .
\ee
The exact totals are
\be
\begin{array}{c|rrrrrrr}
 n&2&3&4&5&6&7&8\\ \hline
 \#\mathfrak g&3&15&20&18&14&5&1\\
 \#\Sigma&18&252&2{,}596&11{,}406&48{,}108&96{,}084&77{,}220
\end{array}
\ee
and hence
\be
 \sum_n\#\mathfrak g=76,\qquad
 \sum_n\#\Sigma=235{,}684\ .
\ee
Table~\ref{tab:76sectors} lists the 76 sectors and the number of
candidate anomaly-vector classes generated in each one.

\begingroup
\small
\setlength{\tabcolsep}{5pt}
\renewcommand{\arraystretch}{1.04}
\begin{longtable}{@{}r r >{\raggedright\arraybackslash}p{0.56\textwidth} r@{}}
\toprule
\#&\(n\)&gauge algebra \(\mathfrak g\)&\(N_{\mathrm{sig}}\)\\
\midrule
\endfirsthead
\toprule
\#&\(n\)&gauge algebra \(\mathfrak g\)&\(N_{\mathrm{sig}}\)\\
\midrule
\endhead
1 & 2 & $A_3\times C_6$ & 6 \\
2 & 2 & $A_3\times E_6$ & 6 \\
3 & 2 & $A_6\times D_5$ & 6 \\
4 & 3 & $A_3^{2}\times A_7$ & 42 \\
5 & 3 & $A_4\times A_6\times B_3$ & 24 \\
6 & 3 & $A_4\times A_6\times C_3$ & 12 \\
7 & 3 & $A_4\times C_5\times G_2$ & 12 \\
8 & 3 & $A_4^{2}\times D_5$ & 18 \\
9 & 3 & $A_5\times A_6\times C_2$ & 36 \\
10 & 3 & $B_3\times B_4\times C_4$ & 16 \\
11 & 3 & $B_3\times B_4^{2}$ & 12 \\
12 & 3 & $B_3\times C_4^{2}$ & 12 \\
13 & 3 & $B_4\times C_3\times C_4$ & 8 \\
14 & 3 & $B_4^{2}\times C_3$ & 6 \\
15 & 3 & $C_2\times C_5\times D_4$ & 18 \\
16 & 3 & $C_3\times C_4^{2}$ & 6 \\
17 & 3 & $A_4\times B_5\times G_2$ & 12 \\
18 & 3 & $B_5\times C_2\times D_4$ & 18 \\
19 & 4 & $A_3\times B_3\times B_4\times C_3$ & 96 \\
20 & 4 & $A_3\times B_3\times C_3\times C_4$ & 96 \\
21 & 4 & $A_3\times B_3^{2}\times B_4$ & 120 \\
22 & 4 & $A_3\times B_3^{2}\times C_4$ & 120 \\
23 & 4 & $A_3\times B_4\times C_3^{2}$ & 36 \\
24 & 4 & $A_3\times B_4\times D_4\times G_2$ & 96 \\
25 & 4 & $A_3\times C_3^{2}\times C_4$ & 36 \\
26 & 4 & $A_3\times C_4\times D_4\times G_2$ & 96 \\
27 & 4 & $A_3^{2}\times A_5\times D_4$ & 84 \\
28 & 4 & $A_3^{3}\times A_6$ & 112 \\
29 & 4 & $A_4\times C_2\times D_5\times G_2$ & 324 \\
30 & 4 & $A_4^{2}\times A_5\times C_2$ & 108 \\
31 & 4 & $A_4^{3}\times B_3$ & 40 \\
32 & 4 & $A_4^{3}\times C_3$ & 20 \\
33 & 4 & $A_6\times B_3\times C_2\times G_2$ & 288 \\
34 & 4 & $A_6\times C_2\times C_3\times G_2$ & 144 \\
35 & 4 & $A_7\times C_2^{3}$ & 330 \\
36 & 4 & $C_2\times C_5\times G_2^{2}$ & 90 \\
37 & 4 & $C_2^{2}\times D_4\times D_5$ & 270 \\
38 & 4 & $B_5\times C_2\times G_2^2$ & 90 \\
39 & 5 & $A_3\times A_6\times C_2^{3}$ & 1,980 \\
40 & 5 & $A_3\times B_4\times G_2^{3}$ & 240 \\
41 & 5 & $A_3\times C_4\times G_2^{3}$ & 240 \\
42 & 5 & $A_3^{2}\times A_5\times G_2^{2}$ & 420 \\
43 & 5 & $A_3^{2}\times B_3\times C_3^{2}$ & 252 \\
44 & 5 & $A_3^{2}\times B_3\times D_4\times G_2$ & 672 \\
45 & 5 & $A_3^{2}\times B_3^{2}\times C_3$ & 420 \\
46 & 5 & $A_3^{2}\times B_3^{3}$ & 420 \\
47 & 5 & $A_3^{2}\times C_3\times D_4\times G_2$ & 336 \\
48 & 5 & $A_3^{2}\times C_3^{3}$ & 84 \\
49 & 5 & $A_3^{3}\times A_4^{2}$ & 336 \\
50 & 5 & $A_4\times A_5\times C_2^{2}\times G_2$ & 1,080 \\
51 & 5 & $A_4\times B_3\times C_2^{2}\times D_4$ & 1,080 \\
52 & 5 & $A_4\times C_2^{2}\times C_3\times D_4$ & 540 \\
53 & 5 & $A_4^{2}\times B_3\times C_2\times G_2$ & 864 \\
54 & 5 & $A_4^{2}\times C_2\times C_3\times G_2$ & 432 \\
55 & 5 & $A_5\times C_2^{3}\times D_4$ & 660 \\
56 & 5 & $C_2^{2}\times D_5\times G_2^{2}$ & 1,350 \\
57 & 6 & $A_3\times A_4^{2}\times C_2^{3}$ & 5,940 \\
58 & 6 & $A_3^{2}\times B_3\times G_2^{3}$ & 1,680 \\
59 & 6 & $A_3^{2}\times C_3\times G_2^{3}$ & 840 \\
60 & 6 & $A_3^{3}\times A_4\times C_2\times G_2$ & 6,048 \\
61 & 6 & $A_3^{3}\times C_2^{2}\times D_4$ & 5,040 \\
62 & 6 & $A_4\times B_3\times C_2^{2}\times G_2^{2}$ & 5,400 \\
63 & 6 & $A_4\times C_2^{2}\times C_3\times G_2^{2}$ & 2,700 \\
64 & 6 & $A_5\times C_2^{3}\times G_2^{2}$ & 3,300 \\
65 & 6 & $B_3\times C_2^{3}\times C_3^{2}$ & 1,980 \\
66 & 6 & $B_3\times C_2^{3}\times D_4\times G_2$ & 5,280 \\
67 & 6 & $B_3^{2}\times C_2^{3}\times C_3$ & 3,300 \\
68 & 6 & $B_3^{3}\times C_2^{3}$ & 3,300 \\
69 & 6 & $C_2^{3}\times C_3\times D_4\times G_2$ & 2,640 \\
70 & 6 & $C_2^{3}\times C_3^{3}$ & 660 \\
71 & 7 & $A_3\times A_4\times C_2^{4}\times G_2$ & 35,640 \\
72 & 7 & $A_3\times C_2^{5}\times D_4$ & 15,444 \\
73 & 7 & $A_3^{3}\times C_2^{2}\times G_2^{2}$ & 25,200 \\
74 & 7 & $B_3\times C_2^{3}\times G_2^{3}$ & 13,200 \\
75 & 7 & $C_2^{3}\times C_3\times G_2^{3}$ & 6,600 \\
76 & 8 & $A_3\times C_2^{5}\times G_2^{2}$ & 77,220 \\
\bottomrule
\caption{The $76$ gauge algebra sectors with non-empty candidate anomaly-vector classes. Here $n$ denotes the number of group factors, and $N_{sig}$ is the number of candidate anomaly-vector classes, totaling 235,684.}
\label{tab:76sectors}\\
\end{longtable}
\endgroup

Directly solving the anomaly equations for the 76 gauge sectors and $235{,}684$ candidate anomaly-vector classes is prohibitively expensive: for each class, the multiplicities of all admissible product irreducible representations become variables in a Diophantine system, and the number of variables can reach the thousands when several low-dimensional simple factors are present. Instead, we use the spectra--assembly method described below, which avoids scanning every candidate anomaly-vector class against every product-representation combination. 

\subsection{The spectra--assembly method}
The key ingredient in this approach, which we call the spectra-assembly method, is the assembly of candidate spectra for an $n$-factor model from the knowledge of lower factor anomaly-free spectra.

We consider six-dimensional $\mathcal N=(1,0)$ supergravity models with gauge group
\be
 {\cal G}_n =G_1\times\cdots\times G_n\times U(1)^M\times \Sp(1)_R\ ,
 \label{G}
\ee
where every hypermultiplet is neutral under the abelian factor $U(1)^M$.  These inert vector multiplets 
will be called \emph{drones}.  The same assumptions concerning the allowed simple groups, representations, multiplicities, singlets, 
and half-hypermultiplets are imposed at every value of $n$.
In the method we are employing, we do not impose the condition $M=0$ from the start. Rather, we find all local and global anomaly-free spectra (possibly with drones), and finally impose the drone-free condition.

\subsubsection*{Projections preserving anomaly-free conditions}

Let $\mathcal M_{n,M}$ be a candidate model with non-abelian factors
$G_1,\ldots,G_{n}$.  Its $i$-th projection is obtained by deleting $G_i$ and 
replacing $\dim G_i$ vector multiplets by the same number of drone vector multiplets,
\be
 {\cal G}_n 
 \longrightarrow
   \pi_i ({\cal G}_n) = \prod_{j\ne i}G_j\times U(1)^{M+\dim G_i}\times\Sp(1)_R.
\ee
The total number of vector multiplets is therefore unchanged.

At the level of matter, the deleted representation index is converted into a multiplicity.  A hypermultiplet
\be
 (R_1,\ldots,R_i,\ldots,R_{n})
 \label{df1}
\ee
contributes 
\be
 \dim R_i\,(R_1,\ldots,\widehat{R_i},\ldots,R_{n})
\ee
to the $i$-th projection. This represents the block \eq{df1} with $R_i$ omitted, and with multiplicity ${\rm dim}\, R_i$.
The additional ${\rm dim}\, G_i$ abelian factors are inert, and all projected fields remain neutral under the full drone group. Because both $n_V$ and $n_H$ are preserved by this operation, 
the projection preserves the anomaly equations involving the surviving
factors. This follows directly from the trace identities. For example,
for a product representation \(R_i\otimes R_j\),
\begin{equation}
  \tr_{R_i\otimes R_j}F_i^2
  =
  (\dim R_j)\,\tr_{R_i}F_i^2 ,
\end{equation}
and similarly for quartic traces. Hence the contribution of the parent
representation to every pure anomaly of a surviving factor is precisely
the contribution of its projection. The same statement holds for the
gravitational and mixed anomalies among the factors that
remain after projection. The symplectic Majorana--Weyl reality condition
is inherited by the projected hyperfermion representation. Accordingly, the projection of any genuine \(n\)-factor model must lie
in the complete \((n-1)\)-factor catalogue used by
the assembly procedure.


The hypermultiplet reality condition is preserved by the projection.  In the full-hypermultiplet convention, a parent multiplicity $m$ can be half-integral only when the full product representation $(R_i\otimes R_{\rm rest})$ is pseudoreal.  If $R_{\rm rest}$ is pseudoreal, the projected block is again allowed to carry a half-integral multiplicity.  If $R_{\rm rest}$ is real, then $R_i$ must be pseudoreal; since every pseudoreal representation has even dimension, the projected multiplicity $(\dim R_i)m$ is integral, as required for a real block.  Complex parent blocks have integral multiplicity from the outset, thus the full- versus half-hypermultiplet rule is inherited by every projection.

We emphasize that the drone-containing projections introduced here are not themselves regarded as admissible physical supergravity models: section~\ref{sec: consistency conditions} shows that neutral abelian vector multiplets have a degenerate gauge-kinetic matrix for $n_T=1$. Their role here is purely auxiliary. We relax the drone-free condition for a moment by temporarily allowing inert abelian vectors so as to preserve $n_V=96$ under deletion of a simple factor. Every physical no-drone $(n+1)$-factor spectrum necessarily projects to such $n$-factor spectra satisfying the same local anomaly, gravitational, representation-reality, and unimodular-lattice constraints. Completeness of the projected lower-factor catalogues therefore supplies necessary constraints on the physical higher-factor search.

\subsection*{Starting point: complete one-factor anomaly-free catalogue}

Even though we are ultimately interested in drone-free models, the recursive
assembly procedure starts from the complete catalogue of one-factor
anomaly-free spectra with drones. For a single non-abelian factor $G$, the
fixed $n_V=96$ vector-multiplet budget requires
$M=93-\dim G$ inert $U(1)$ vector multiplets, so that
$G\times U(1)^M\times Sp(1)_R$ has the required total number of vector
multiplets. Solving the anomaly equations \eqref{eqn: one-factor anomaly equations} in a straightforward way using the anomaly vectors listed in Table \ref{tab:local48}, one can derive the complete list of one-factor anomaly-free spectra. The resulting one-factor catalogue is summarized in
Table~\ref{tab:n1-one-factor-catalogue}.

\begin{table}[H]
\centering
\begin{tabular}{lrr}
\toprule
Group & Branches & Spectra \\
\midrule
$A_3$ & 3 & 187 \\
$A_4$ & 2 & 12 \\
$A_5$ & 2 & 44 \\
$A_6$ & 1 & 7 \\
$A_7$ & 2 & 5 \\
$B_3$ & 3 & 8 \\
$B_4$ & 2 & 4 \\
$C_2$ & 7 & 279 \\
$C_3$ & 2 & 43 \\
$C_4$ & 2 & 12 \\
$C_5$ & 1 & 2 \\
$C_6$ & 1 & 1 \\
$D_4$ & 2 & 27 \\
$D_5$ & 3 & 3 \\
$D_7$ & 1 & 1 \\
$G_2$ & 3 & 10 \\
$E_6$ & 1 & 1 \\
\bottomrule
\end{tabular}
\caption{The one-factor anomaly-free spectra with drones used as the seed
catalogue for the spectra--assembly method.  The ``Branches''
column counts the distinct one-factor anomaly-vector classes for each simple gauge
factor, totaling 38, while ``Spectra'' gives the number of spectra in those branches, summing up to 646.}
\label{tab:n1-one-factor-catalogue}
\end{table}

Altogether, the one-factor seed catalogue contains $646$ spectra distributed
over $38$ anomaly branches.  These drone-containing spectra are not retained
as acceptable final theories; rather, they provide the lower-factor data from
which higher-factor spectra are assembled.  In particular, a candidate at the
next level can survive only if each of its projections appears in
the corresponding lower-factor catalogue.

\begingroup
\subsection*{Constructing two-and-higher factor catalogues}
\label{subsec: two-factor catalogue}
Given a pair of one-factor anomaly-free spectra for $G_1$ and $G_2$,
\begin{equation}
\left\{
\begin{aligned}
& m_1\mathbf{R}_1 \oplus m_2\mathbf{R}_2 \oplus \cdots \oplus m_t\mathbf{R}_t \\
& m'_1\mathbf{R}'_1 \oplus m'_2\mathbf{R}'_2 \oplus \cdots \oplus m'_s\mathbf{R}'_s,
\end{aligned}
\right.
\label{eqn:pair of one-factor spectra}
\end{equation}
a two-factor parent spectrum with these prescribed one-factor projections has the general form\footnote{In the language of \cite{Becker:2023zyb}, this means there exist links connecting these two nodes.}
\begin{equation}
    m_{11}(\mathbf{R}_1,\mathbf{R}'_1)
    \oplus m_{12}(\mathbf{R}_1,\mathbf{R}'_2)
    \oplus \cdots
    \oplus m_{1s}(\mathbf{R}_1,\mathbf{R}'_s)
    \oplus m_{21}(\mathbf{R}_2,\mathbf{R}'_1)
    \oplus \cdots
    \oplus m_{ts}(\mathbf{R}_t,\mathbf{R}'_s),
    \label{eqn:n=2 general form}
\end{equation}
where $m_{ij}$ are the physical multiplicities of the corresponding irreps and obey the full-hypermultiplet reality rule, the unprimed irreps are the ones associated to $G_1$ while the primed irreps are associated to $G_2$.

Projecting the two-factor model into a pair of one-factor models, whose spectra take the form \eq{eqn:pair of one-factor spectra} with 
\begin{equation}
\left\{
\begin{aligned}
&m_1=m_{11}\mathrm{dim}\mathbf{R}'_1+\cdots+m_{1s}\mathrm{dim}\mathbf{R}'_s,
~~\cdots~~,
m_t=m_{t1}\mathrm{dim}\mathbf{R}'_1+\cdots+m_{ts}\mathrm{dim}\mathbf{R}'_s,\\
&m'_1=m_{11}\mathrm{dim}\mathbf{R}_1+\cdots+m_{t1}\mathrm{dim}\mathbf{R}_t,
~~\cdots~~,
m'_s=m_{1s}\mathrm{dim}\mathbf{R}_1+\cdots+m_{ts}\mathrm{dim}\mathbf{R}_t.
\end{aligned}
\right.
\label{eqn:solve for x_ij}
\end{equation}
Solving the $t+s$ equations \eqref{eqn:solve for x_ij} for the physical variables $m_{ij}$ gives the {\it candidate} $n=2$ parent spectra that reduce to the two prescribed $n=1$ spectra upon projection onto either gauge factor.\footnote{For the exact integer implementation one sets $X_{ij}=2m_{ij}$.  The $X_{ij}$ are nonnegative integers, even for real or complex parent blocks and unrestricted in parity for pseudoreal parent blocks.}  Satisfying \eqref{eqn:solve for x_ij} alone is not sufficient to guarantee that a candidate spectrum is anomaly-free. One must additionally impose the mixed-anomaly condition \eq{eq:general-local-anomaly-equations}
\begin{equation}
b_i\cdot b_j=\lambda_i\lambda_j M_{ij}.
\end{equation}
The quantity $M_{ij}$ on the right-hand side can be computed for the candidate two-factor spectrum given by the solution $m_{ij}$, while the left-hand side is already fixed by the anomaly coefficients of the chosen pair of one-factor spectra. A candidate therefore defines a formal anomaly-free $n=2$ spectrum if and only if the two sides agree.

Once the complete set of local anomaly-free $n=1$ spectra for all simple gauge groups is known, all local anomaly-free $n=2$ spectra can in principle be generated by pairing the one-factor spectra and imposing the corresponding consistency conditions. The construction can then be iterated: by combining $n=1$ and $n=2$ spectra, one can generate the $n=3$ catalogue, and similarly proceed to spectra with an arbitrary number of gauge factors. We provide the full catalogue of these formal anomaly-free spectra (with drones) as part of the ancillary materials\cite{github}, in which the gauge group, anomaly-vector classes, number of drones and matter spectra are included.

\subsection*{Comparison with ``brute-force'' enumeration}

It is useful to compare the size of this assembly problem with a direct
(``brute-force'') solution of the anomaly equations.  For a fixed
\(n\)-factor gauge algebra, let \(r_i\) denote the number of irreps of
\(G_i\) retained in the representation menu (including the trivial
representation when it is used to describe matter charged only under
the other factors).  A direct product-representation ansatz can then
contain, before any anomaly equation is imposed, as many as
\begin{equation}
  V_{\rm brute}\sim \prod_{i=1}^n r_i
  \label{eq:bruteforce-variable-count}
\end{equation}
independent multiplicity variables.  The precise number is smaller when
representation-reality restrictions or dimension bounds remove product
irreps, but the multiplicative growth with \(n\) is the essential
point.  Even for a modest common menu size \(r_i\sim6\)--\(10\), this
corresponds parametrically to \(O(10^2)\) variables for \(n=2\),
\(O(10^3\!-\!10^4)\) for \(n=4\), and
\(O(10^5\!-\!10^8)\) for \(n=6\)--\(8\).

For a fixed anomaly-vector assignment, these variables are constrained
by the gravitational equation, the pure-factor \(A_i,B_i,C_i\)
conditions (and a Pfaffian condition when \(G_i=D_4\)), and all pairwise
mixed-anomaly equations.  Ignoring the occasional Pfaffian equation,
the nominal number of equations, before removing identically trivial or dependent ones, is therefore
\begin{equation}
 E_{\rm brute}
 =
 1+3n+\binom{n}{2},
 \label{eq:bruteforce-equation-count}
\end{equation}
namely \(8,13,19,34,53\) equations for
\(n=2,3,4,6,8\), respectively.  The difficulty is not primarily the
number of equations, which grows only quadratically, but the rapidly
growing number of integer variables and the need to repeat
the calculation over the \(235{,}684\) candidate anomaly-vector
assignments of Table~\ref{tab:76sectors}.

The spectra--assembly problem is much more constrained.  For example,
let two selected one-factor spectra contain \(t\) and \(s\) occupied
irreps.  Their most general two-factor parent contains at most
\begin{equation}
  V_{\rm ass}=ts
\end{equation}
physical link multiplicities \(m_{ij}\).  These variables are not free:
the given single-factor-spectra multiplicities impose the \(t+s\) linear equations \eqref{eqn:solve for x_ij}, after which only the
mixed-anomaly condition has to be checked explicitly.  Thus the pure
\(G_1\), pure \(G_2\), gravitational, representation-reality, and
one-factor lattice constraints have already been solved once and for
all in the seed catalogue; they are not re-solved for every parent model.
Moreover, \(s\) and \(t\) count only irreps actually occupied in the
chosen projections and are therefore generally much smaller than the full
menu sizes \(r_1,r_2\).  The matching equations also give immediate
upper bounds on every \(m_{ij}\), which sharply reduces the reconstruction
search.

There are
\begin{equation}
 \frac{646\times647}{2}=208{,}981
\end{equation}
unordered pairs of entries in the complete one-factor seed catalogue,
but most are rejected by group, branch, dimension, or pair-compatibility
conditions before a nontrivial integer reconstruction is attempted.
At higher factor number the gain is more pronounced: the direct method
starts anew from an \(n\)-fold product-representation menu, whereas the
assembly method introduces only links compatible with already solved
lower-factor spectra.  In this sense the method replaces one very large
Diophantine system for each global anomaly-vector candidate by a large
number of small, strongly constrained reconstruction problems.
This is the principal computational reason that the exhaustive
classification is tractable in Mathematica.

\endgroup

\subsection{Completeness of the spectra--assembly construction}
\label{sec:shadow-completeness}

We further clarify why under the assumptions used throughout the scan, the
spectra--assembly construction is exhaustive rather than merely a search
heuristic. The argument is inductive in the number \(n\) of nonabelian
simple factors. 

\subsection*{The two-factor case}

The logical content of the completeness argument is clearest for
\(n=2\). Let
\begin{equation}
  \mathcal M:
  \qquad
  G_1\times G_2
  \times U(1)^M_{\rm drone}
  \times Sp(1)_R
  \label{eq:two-factor-parent}
\end{equation}
be any formal anomaly-free two-factor model. Its two one-factor projections are
\begin{equation}
  S_1=\pi_2(\mathcal M),
  \qquad
  S_2=\pi_1(\mathcal M).
  \label{eq:two-shadows}
\end{equation}
If the one-factor catalogue is complete, then
\begin{equation}
  S_1,\ S_2 \in \mathscr C_1 ,
\end{equation}
where \(\mathscr C_1\) denotes that catalogue.

The parent spectrum itself supplies a solution of the exact assembly
equations for the pair \((S_1,S_2)\). Indeed, the multiplicities of
bifundamental and singly charged representations in \(\mathcal M\)
reconstruct both projected spectra under \(\pi_1\) and \(\pi_2\),
and they satisfy the required local anomaly, gravitational, and
representation-theoretic constraints by hypothesis.

Suppose now, for contradiction, that there exists a formal anomaly-free
two-factor model \(\mathcal M_\ast\) which is absent from the output of
the spectra--assembly scan. Its projections
\begin{equation}
  S_{1,\ast}=\pi_2(\mathcal M_\ast),
  \qquad
  S_{2,\ast}=\pi_1(\mathcal M_\ast)
\end{equation}
must belong to the complete one-factor catalogue. But the actual
multiplicities of \(\mathcal M_\ast\) provide an admissible integer
solution of the assembly problem for the compatible pair
\((S_{1,\ast},S_{2,\ast})\). Therefore an exhaustive solution of that
assembly problem must reproduce \(\mathcal M_\ast\), contradicting the
assumption that it was missing.

The crucial point is that uniqueness is not required. Several distinct
two-factor parents may have the same pair of projections. Completeness
requires only that the assembly step enumerate the entire fiber
\begin{equation}
  \pi^{-1}(S_1,S_2)
\end{equation}
for every compatible projection pair.

\paragraph{Completeness proposition.}
The construction described above has the following properties:
\begin{enumerate}
  \item the bounded one-factor systems are solved exhaustively over the
        stated representation domain, yielding the complete base catalogue
        \(\mathscr C_1\);
  \item every projection of an anomaly-free \(n\)-factor model onto an
        \((n-1)\)-factor spectrum preserves all conditions defining membership
        in \(\mathscr C_{n-1}\);
  \item for every compatible tuple of \((n-1)\)-factor spectra, the
        assembly equations are solved exhaustively over all allowed integer
        and half-integer hypermultiplet multiplicities, with no truncation
        beyond the stated scan assumptions.
\end{enumerate}
These properties imply that the spectra--assembly construction is
complete for every \(n\): every model in \(\mathscr C_n\) appears in the
assembled catalogue.

\paragraph{Proof.}
The base case \(n=1\) is established by the exhaustive one-factor scan:
by construction, every admissible one-factor spectrum is contained in
\(\mathscr C_1\). Suppose now that the construction is complete for
\(n-1\), and let
\begin{equation}
  \mathcal M\in\mathscr C_n
\end{equation}
be any anomaly-free \(n\)-factor model. Projecting out each simple factor
gives the \(n\)-tuple of projected spectra
\begin{equation}
  \bigl(
  \pi_1(\mathcal M),\ldots,\pi_n(\mathcal M)
  \bigr).
  \label{eq:shadow-tuple}
\end{equation}
By the projection property and the induction hypothesis, every member
of this tuple belongs to the complete \((n-1)\)-factor catalogue.

The multiplicities of the parent model \(\mathcal M\) furnish a solution
of the assembly equations associated with the compatible tuple
\eqref{eq:shadow-tuple}. Since the assembly step enumerates all such
solutions, \(\mathcal M\) must occur among the reconstructed
\(n\)-factor models. Hence no anomaly-free \(n\)-factor model lies outside
the assembled catalogue. This establishes the induction step and,
together with the exhaustive one-factor scan, proves completeness for
all \(n\).

The accompanying Mathematica implementation realizes each stage of this induction using exact integer arithmetic; its resulting catalogue is summarized in Table~\ref{tab:nV96-counting} and supplied as ancillary material.

\subsection{Summary of results}
In the enlarged search space that allows drone vector multiplets, the spectra--assembly scan yields \(4{,}094\) formal anomaly- and lattice-consistent spectra. Their distribution according to the number \(n\) of simple non-\(R\) gauge factors is summarized
in Table~\ref{tab:nV96-counting}.

\begin{table}[ht]
\centering
\begin{tabular}{c r}
\hline
Number of simple non-\(R\) factors \(n\)
& Number of spectra \\
\hline
\(1\)        & \(646\)   \\
\(2\)        & \(1{,}968\) \\
\(3\)        & \(1{,}092\) \\
\(4\)        & \(388\)   \\
\(n\geq 5\)  & \(0\)     \\
\hline
Total        & \(4{,}094\) \\
\hline
\end{tabular}
\caption{Numbers of formal anomaly- and lattice-consistent spectra obtained from the
spectra--assembly scan, organized by the number of simple non-\(R\)
gauge factors.}
\label{tab:nV96-counting}
\end{table}

Note that the five-factor scan is empty. Since any model with more than five simple
factors would have a nonempty five-factor projection, this also excludes all
\(n>5\) spectra. Moreover, every one of the \(4{,}094\) spectra contains at least one
drone vector multiplet. Thus, after imposing the drone-free condition in the end, the \(n_V=96\) no-drone semisimple branch does not yield any consistent model, assuming that  the non-\(R\) gauge algebra contains no \(U(1)\), \(A_1\) and \(A_2\) factors.

\section{The $A_1^3 \times Sp(1)_R$ models}
\label{sec:nV12}

\subsection{Lattice reconstruction and the 120 candidate classes}
\label{sec:a1-classes}

In this $n_V=12$ branch one has $n_H=256$. For simplicity, we restrict the matter representation menu to singlets, doublets, and triplets of each \(A_1\) factor, with at most one triplet allowed in any product representation. In other words, the charged representation domain consists of blocks of the form
\begin{equation}
\mathcal R_\mu=(d_{\mu1},d_{\mu2},d_{\mu3}),
\qquad
d_{\mu i}\in\{1,2,3\},
\qquad
\sum_{i=1}^{3}\delta_{d_{\mu i},\,3}\le1 .
\label{eq:A1x3-menu}
\end{equation}
We use the off-diagonal presentation of the unimodular string-charge lattice,
\begin{equation}
 \eta=\begin{pmatrix}0&1\\[1mm]1&0\end{pmatrix},
 \qquad a=(-2,-2), \qquad b_R=(-2,2), \qquad n_H=256,
 \label{eq:a1-background-data}
\end{equation}
so that $x\cdot y=x_1y_2+x_2y_1$.  For each $A_1$ factor define
\begin{equation}
 p_i:=a\cdot b_i,
 \qquad
 q_i:=b_R\cdot b_i.
 \label{eq:a1-pq-definition}
\end{equation}
The mixed $Sp(1)_R$ anomaly fixes $q_i=-2$.  Writing $b_i=(u_i,v_i)$,
the two scalar products therefore give
\begin{equation}
 -2(u_i+v_i)=p_i,
 \qquad
 2(u_i-v_i)=-2.
 \label{eq:a1-reconstruction-equations}
\end{equation}
Putting $r_i:=u_i$, integral lattice reconstruction yields
\begin{equation}
 b_i=(r_i,r_i+1),
 \qquad
 p_i=-4r_i-2,
 \qquad
 q_i=-2,
 \qquad
 b_i^2=2r_i(r_i+1),
 \qquad
 r_i\in\mathbb Z.
 \label{eq:a1-ri-parametrization}
\end{equation}
The active one-factor representation menu contains the doublet
$\mathbf2$ and the adjoint $\mathbf3$, whose trace coefficients are
\begin{equation}
 (A_{\mathbf2},C_{\mathbf2})=
 \left(1,\frac12\right),
 \qquad
 (A_{\mathbf3},C_{\mathbf3})=(4,8).
 \label{eq:a1-trace-coefficients}
\end{equation}
Let $m_{2,i}$ and $m_{3,i}$ denote the effective physical multiplicities seen by
the $i$th factor, in full-hypermultiplet units.  The one-factor
anomaly targets are
\begin{equation}
 \widehat A_i=4-6p_i=16+24r_i,
 \qquad
 \widehat C_i=8+3b_i^2=8+6r_i(r_i+1).
 \label{eq:a1-local-targets}
\end{equation}

The charged representation domain is as described in \eq{eq:A1x3-menu}. 
Let $m_\mu$ denote the physical multiplicity of the block
$\mathcal R_\mu$, measured in full-hypermultiplet units.  Because the $A_1$
doublet is pseudoreal whereas the singlet and triplet are real, the full
product representation is pseudoreal precisely when it contains an odd
number of doublet factors.  The physical SMW rule is therefore
\begin{equation}
 m_\mu\in\frac12\mathbb Z_{\geq0}
 \quad\text{if $\mathcal R_\mu$ is pseudoreal},
 \qquad
 m_\mu\in\mathbb Z_{\geq0}
 \quad\text{if $\mathcal R_\mu$ is real}.
 \label{eq:a1-full-smw-rule}
\end{equation}
If $s$ denotes the number of neutral full hypermultiplets, the gravitational
condition is
\begin{equation}
 \sum_\mu \dim(\mathcal R_\mu)m_\mu+s=256,
 \qquad
 s\in\mathbb Z_{\geq0}.
 \label{eq:a1-gravitational-equation}
\end{equation}

For later use, define the effective physical doublet and triplet contents of
the $i$-th factor by
\begin{align}
 D_i&:=\sum_{\mu:\,d_{\mu i}=2}
 m_\mu\prod_{\ell\neq i}d_{\mu\ell},
 \label{eq:a1-Di-definition}\\
 T_i&:=\sum_{\mu:\,d_{\mu i}=3}
 m_\mu\prod_{\ell\neq i}d_{\mu\ell}.
 \label{eq:a1-Ti-definition}
\end{align}
Consequently the anomaly equations \eqref{eq:anomaly-targets} are
\begin{align}
D_i+4T_i&=\widehat A_i=16+24r_i,
 \label{eq:a1-local-A-equation}\\
 \frac12D_i+8T_i&=\widehat C_i
 =8+6r_i(r_i+1).
 \label{eq:a1-local-C-equation}
\end{align}
Solving them gives
\begin{equation}
T_i=r_i(r_i-1),
\qquad
D_i=4\bigl(4+7r_i-r_i^2\bigr).
 \label{eq:a1-effective-multiplicities}
\end{equation}
To derive the candidate anomaly-vector classes, we can first restrict ourselves to the single-factor case. Then the nonnegativity of LHS of Equation \eqref{eq:a1-effective-multiplicities} together with the necessary dimension bound $n_H=256$, implies
\begin{equation}
2D_i+3T_i\leq256, \qquad
 r_i\in\{0,1,2,3,4,5,6,7\}.
 \label{eq:a1-eight-local-classes}
\end{equation}
For the multi-factor case, one can use the idea of spectra-assembly illustrated in Section \ref{sec: nV=96 branch} to show that no additional anomaly-vector classes can arise beyond those identified above.

The mixed equation may be written as
\begin{equation}
 P_{ij}:=\sum_\mu m_\mu
 A_{d_{\mu i}}A_{d_{\mu j}}
 \prod_{\ell\neq i,j}d_{\mu\ell}
 =b_i\cdot b_j,
 \qquad
 (A_1,A_2,A_3)=(0,1,4).
 \label{eq:a1-Pij-definition}
\end{equation}

Since the three $A_1$ factors are identical, a candidate global class has
a unique ordered representative
\begin{equation}
 0\leq r_1\leq r_2\leq r_3\leq7.
 \label{eq:a1-ordered-triple}
\end{equation}
The number of such multisets is
\begin{equation}
 N_{\mathrm{candidate}}=\binom{8+3-1}{3}=120.
 \label{eq:a1-120-count}
\end{equation}
Table~\ref{tab:a1-all-120} lists them.  A dagger marks a class which does not admit matter spectra, which will be proved in section~\ref{sec:a1-empty-fibers}.

\begin{table}[p]
\centering
\begingroup
\scriptsize
\setlength{\tabcolsep}{2pt}
\renewcommand{\arraystretch}{1.05}

\begin{tabular}{*{5}{r c @{\hspace{0.10in}}} r c}
\toprule
\multicolumn{12}{c}{\textit{Nonempty classes}}\\
\midrule
\# & $(r_1,r_2,r_3)$
& \# & $(r_1,r_2,r_3)$
& \# & $(r_1,r_2,r_3)$
& \# & $(r_1,r_2,r_3)$
& \# & $(r_1,r_2,r_3)$
& \# & $(r_1,r_2,r_3)$\\
\midrule
1  & $(0,0,0)$ & 18 & $(0,2,4)$ & 35 & $(1,1,1)$
& 52 & $(1,3,7)$ & 69 & $(2,3,5)$ & 86  & $(3,4,5)$ \\

2  & $(0,0,1)$ & 19 & $(0,2,5)$ & 36 & $(1,1,2)$
& 53 & $(1,4,4)$ & 70 & $(2,3,6)$ & 87  & $(3,4,6)$ \\

3  & $(0,0,2)$ & 20 & $(0,2,6)$ & 37 & $(1,1,3)$
& 54 & $(1,4,5)$ & 71 & $(2,3,7)$ & 88  & $(3,4,7)$ \\

4  & $(0,0,3)$ & 21 & $(0,2,7)$ & 38 & $(1,1,4)$
& 55 & $(1,4,6)$ & 72 & $(2,4,4)$ & 89  & $(3,5,5)$ \\

5  & $(0,0,4)$ & 22 & $(0,3,3)$ & 39 & $(1,1,5)$
& 56 & $(1,4,7)$ & 73 & $(2,4,5)$ & 90  & $(3,5,6)$ \\

6  & $(0,0,5)$ & 23 & $(0,3,4)$ & 40 & $(1,1,6)$
& 57 & $(1,5,5)$ & 74 & $(2,4,6)$ & 91  & $(3,5,7)$ \\

7  & $(0,0,6)$ & 24 & $(0,3,5)$ & 41 & $(1,1,7)$
& 58 & $(1,5,6)$ & 75 & $(2,4,7)$ & 92  & $(3,6,6)$ \\

8  & $(0,0,7)$ & 25 & $(0,3,6)$ & 42 & $(1,2,2)$
& 59 & $(1,5,7)$ & 76 & $(2,5,5)$ & 93  & $(4,4,4)$ \\

9  & $(0,1,1)$ & 26 & $(0,3,7)$ & 43 & $(1,2,3)$
& 60 & $(1,6,6)$ & 77 & $(2,5,6)$ & 94  & $(4,4,5)$ \\

10 & $(0,1,2)$ & 27 & $(0,4,4)$ & 44 & $(1,2,4)$
& 61 & $(2,2,2)$ & 78 & $(2,5,7)$ & 95  & $(4,4,6)$ \\

11 & $(0,1,3)$ & 28 & $(0,4,5)$ & 45 & $(1,2,5)$
& 62 & $(2,2,3)$ & 79 & $(2,6,6)$ & 96  & $(4,4,7)$ \\

12 & $(0,1,4)$ & 29 & $(0,4,6)$ & 46 & $(1,2,6)$
& 63 & $(2,2,4)$ & 80 & $(3,3,3)$ & 97  & $(4,5,5)$ \\

13 & $(0,1,5)$ & 30 & $(0,4,7)$ & 47 & $(1,2,7)$
& 64 & $(2,2,5)$ & 81 & $(3,3,4)$ & 98  & $(4,5,6)$ \\

14 & $(0,1,6)$ & 31 & $(0,5,5)$ & 48 & $(1,3,3)$
& 65 & $(2,2,6)$ & 82 & $(3,3,5)$ & 99  & $(4,5,7)$ \\

15 & $(0,1,7)$ & 32 & $(0,5,6)$ & 49 & $(1,3,4)$
& 66 & $(2,2,7)$ & 83 & $(3,3,6)$ & 100 & $(4,6,6)$ \\

16 & $(0,2,2)$ & 33 & $(0,5,7)$ & 50 & $(1,3,5)$
& 67 & $(2,3,3)$ & 84 & $(3,3,7)$ & 101 & $(5,5,5)$ \\

17 & $(0,2,3)$ & 34 & $(0,6,6)$ & 51 & $(1,3,6)$
& 68 & $(2,3,4)$ & 85 & $(3,4,4)$ & 102 & $(5,5,6)$ \\
\bottomrule
\end{tabular}

\vspace{0.65em}

\begin{tabular}{*{5}{r c @{\hspace{0.10in}}} r c}
\toprule
\multicolumn{12}{c}{\textit{Empty classes}}\\
\midrule
\# & $(r_1,r_2,r_3)$
& \# & $(r_1,r_2,r_3)$
& \# & $(r_1,r_2,r_3)$
& \# & $(r_1,r_2,r_3)$
& \# & $(r_1,r_2,r_3)$
& \# & $(r_1,r_2,r_3)$\\
\midrule
103 & $(0,6,7)^\dagger$
& 106 & $(1,7,7)^\dagger$
& 109 & $(3,6,7)^\dagger$
& 112 & $(4,7,7)^\dagger$
& 115 & $(5,6,7)^\dagger$
& 118 & $(6,6,7)^\dagger$ \\

104 & $(0,7,7)^\dagger$
& 107 & $(2,6,7)^\dagger$
& 110 & $(3,7,7)^\dagger$
& 113 & $(5,5,7)^\dagger$
& 116 & $(5,7,7)^\dagger$
& 119 & $(6,7,7)^\dagger$ \\

105 & $(1,6,7)^\dagger$
& 108 & $(2,7,7)^\dagger$
& 111 & $(4,6,7)^\dagger$
& 114 & $(5,6,6)^\dagger$
& 117 & $(6,6,6)^\dagger$
& 120 & $(7,7,7)^\dagger$ \\
\bottomrule
\end{tabular}

\endgroup
\medskip
\caption{The 120 candidate anomaly-vector classes for
\(A_1^3\times Sp(1)_R\). The first 102 entries are the nonempty classes,
numbered in the same order as the representative-spectrum table; the
final eighteen daggered entries do not admit corresponding spectra.}
\label{tab:a1-all-120}
\end{table}

\subsection{Analytic proof of the eighteen empty fibers}
\label{sec:a1-empty-fibers}

We now derive two necessary inequalities for a nonnegative solution of the
full matter system.  Since the argument only uses $m_\mu\geq0$, it applies in
particular to the full SMW domain of Eq.~\eqref{eq:a1-full-smw-rule}.

\begin{lemma}[Pairwise bound]
Every spectrum in the representation domain
\eqref{eq:A1x3-menu} satisfies
\begin{equation}
 3P_{ij}\leq4(D_i+D_j)
 \label{eq:a1-pairwise-bound}
\end{equation}
for every pair $i\neq j$.
\end{lemma}

\begin{proof}
It suffices to check one product representation at a time.  A nonzero
contribution to $P_{ij}$ can only come from a doublet--doublet or an
adjoint--doublet pair, because two adjoints in one product representation
are excluded.  If $d_k$ is the dimension under the remaining factor, a
doublet--doublet block contributes
\begin{equation}
 \Delta P_{ij}=m_\mu d_k,
 \qquad
 \Delta D_i+\Delta D_j=4m_\mu d_k.
 \label{eq:a1-DD-contribution}
\end{equation}
An adjoint--doublet block contributes
\begin{equation}
 \Delta P_{ij}=4m_\mu d_k,
 \qquad
 \Delta D_i+\Delta D_j=3m_\mu d_k.
 \label{eq:a1-AD-contribution}
\end{equation}
Both obey $3\Delta P_{ij}\leq4(\Delta D_i+\Delta D_j)$.
Blocks with a singlet under either member of the pair have
$\Delta P_{ij}=0$.  Summing the blockwise inequality with nonnegative
coefficients $m_\mu$ proves Eq.~\eqref{eq:a1-pairwise-bound}.
\end{proof}

\begin{lemma}[Three-pair bound]
Every spectrum in the same domain satisfies
\begin{equation}
 12(P_{12}+P_{13}+P_{23})
 \leq19(D_1+D_2+D_3).
 \label{eq:a1-global-bound-P}
\end{equation}
\end{lemma}

\begin{proof}
Again the inequality can be checked block by block.  Up to permutations, the
only relevant representation types are
\begin{equation}
 (2,1,1),\quad(2,2,1),\quad(2,2,2),\quad
 (3,1,1),\quad(3,2,1),\quad(3,2,2).
 \label{eq:a1-block-types}
\end{equation}
When $\sum_i\Delta D_i\neq0$, the largest value of
$\sum_{i<j}\Delta P_{ij}/\sum_i\Delta D_i$ occurs for a permutation of
$(3,2,2)$.  For this block,
\begin{equation}
 \Delta P_{12}+\Delta P_{13}+\Delta P_{23}=19m_\mu,
 \qquad
 \Delta D_1+\Delta D_2+\Delta D_3=12m_\mu.
 \label{eq:a1-322-ratio}
\end{equation}
The remaining types give a smaller ratio; types with vanishing denominator
also have vanishing numerator.  Thus
\begin{equation}
 \sum_{i<j}\Delta P_{ij}
 \leq\frac{19}{12}\sum_i\Delta D_i.
 \label{eq:a1-blockwise-global-bound}
\end{equation}
Summing over all occupied blocks proves Eq.~\eqref{eq:a1-global-bound-P}.
\end{proof}

We first apply the pairwise bound.  From Eq.~\eqref{eq:a1-effective-multiplicities},
\begin{equation}
 D_{r_i=6}=40,
 \qquad
 D_{r_i=7}=16,
 \qquad
 P_{r_i=6,r_j=7}=97,
 \qquad
 P_{r_i=7,r_j=7}=112.
 \label{eq:a1-pairwise-special-values}
\end{equation}
Hence
\begin{align}
 3P_{r_i=6,r_j=7}=291&>224=4(D_{r_i=6}+D_{r_j=7}),
 \label{eq:a1-67-violation}\\
 3P_{r_i=7,r_j=7}=336&>128=4(D_{r_i=7}+D_{r_j=7}).
 \label{eq:a1-77-violation}
\end{align}
Every ordered triple containing a pair $(r_i=6,r_j=7)$ or $(r_i=7,r_j=7)$ is therefore
empty.  In the \((r_i,r_j,r_k)\) labeling, these are
precisely the following fifteen classes:
\begin{align}
 &(0,6,7),(0,7,7),(1,6,7),(1,7,7),(2,6,7),(2,7,7),
 \notag\\
 &(3,6,7),(3,7,7),(4,6,7),(4,7,7),(5,6,7),(5,7,7),
 \notag\\
 &(6,6,7),(6,7,7),(7,7,7).
 \label{eq:a1-fifteen-empty}
\end{align}

The three-pair bound excludes three further classes.  The relevant data are
\begin{equation}
 D_{r_i=5}=56,
 \qquad
 D_{r_i=6}=40,
 \qquad
 D_{r_i=7}=16.
 \label{eq:a1-D567}
\end{equation}
Substitution into Eq.~\eqref{eq:a1-global-bound-P} gives the following exact
violations, in the
$(r_i, r_j , r_k)$ labeling.

\begin{table}[htbp]
\centering
\begin{tabular}{c r r r}
\toprule
class & $\sum_{i<j}P_{ij}$ & $12\sum_{i<j}P_{ij}$
& $19\sum_iD_i$\\
\midrule
$(5,5,7)$ & 224 & 2688 & 2432\\
$(5,6,6)$ & 226 & 2712 & 2584\\
$(6,6,6)$ & 252 & 3024 & 2280\\
\bottomrule
\end{tabular}
\medskip
\caption{The three additional violations of the necessary three-pair bound
\eqref{eq:a1-global-bound-P}.}
\label{tab:a1-global-violations}
\end{table}

We have therefore proved that the empty classes are the fifteen in
Eq.~\eqref{eq:a1-fifteen-empty}, together with
\begin{equation}
 (5,5,7),
 \qquad
 (5,6,6),
 \qquad
 (6,6,6).
 \label{eq:a1-three-additional-empty}
\end{equation}
They are the eighteen daggered entries of Table~\ref{tab:a1-all-120}.

\subsection{Scan strategy and outcome}
\label{sec:a1-outcome}

The analytic inequalities leave 102 candidate classes. The exact scan described below finds at least one spectrum in each of them, so all 102 are nonempty.
Because the representation truncation leaves only \(19\) charged product
irreps, the overhead of the spectra--assembly method provides little advantage
in the \(A_1^3\) branch.  We therefore simply use the search strategy illustrated in Section \ref{sec: a search strategy}. For each fixed ordered class
\(r=(r_1,r_2,r_3)\), the definitions of \(D_i\), \(T_i\), and
\(P_{ij}\) reduce the anomaly equations \eqref{eq:anomaly-targets} to a
\(9\times19\) linear system
$Dm=t(r)$, supplemented by the gravitational condition
\(\sum_{\mu=1}^{19}\dim( R_\mu)m_\mu+s=256\).
Solving these equations for the allowed multiplicities gives all anomaly-free spectra, which will be finally quotiented by permutations
of the identical \(A_1\) factors that preserve the ordered anomaly-vector
class.
After quotienting the stabilizer of its ordered class, the
class-by-class counts give
\begin{equation}
 N_{\mathrm{inequivalent}}=4{,}337{,}331.
 \label{eq:a1-spectrum-counts}
\end{equation}

We then impose the ghost-free condition. All candidate $A_1^3$ anomaly-vector classes admit a common region in tensor-moduli space with positive gauge kinetic terms. In the off-diagonal presentation of U, write $j = (j_1,j_2)$ and choose $j_1 > j_2 > 0$, with an irrelevant overall normalization chosen so that $j$ is a unit positive vector. Then
\be
j \cdot  b_R = 2(j_1-j_2) > 0\ ,
\ee
while, for $b_i = (r_i,r_i+1)$ we have 
\be
j \cdot b_i = (r_i+1)j_1+r_i j_2 > 0\ ,
\ee
for every $r_i \in \{0,...,7\}$. Thus the $Sp(1)_R$ and all three $A_1$ gauge kinetic terms can be positive simultaneously for every one of the $120$ candidate classes, and hence for all $102$ nonempty classes.

Finally, one can show by direct computation that $a=(-2,-2)$ is a characteristic element of the even charge lattice\footnote{Note that all spectra in this section admit an even lattice embedding, but this does not mean that an odd lattice embedding is impossible.}.

\section{The $A_2\times U(1)\times Sp(1)_R$ models}
\label{sec:a2u1}

Allowing hypermultiplets to be charged under abelian factors gives another
dimension-nine non-\(R\) group,
\[
A_2\times U(1).
\]
Here we consider
\begin{equation}
  G=A_2\times U(1)\times Sp(1)_R
\end{equation}
and assume that not all fields are neutral under the ordinary \(U(1)\),
so that this factor is not a drone. To obtain a finite
exhaustive problem, we restrict the \(A_2\) representations and the
charges they can carry to
\begin{equation}
\mathcal R_1=
\{\mathbf 1,\mathbf 3,\overline{\mathbf 3},
\mathbf 6,\overline{\mathbf 6},\mathbf 8\},
\qquad
|q|\leq2 .
\label{eq:a2u1-menu}
\end{equation}
The charge assignment is called primitive when
\begin{equation}
 \gcd\{|q|:\text{a nonzero-charge block is occupied}\}=1.
\label{eq:primitive}
\end{equation}
Consequently, a spectrum supported only at \(|q|=2\) is a nonprimitive
presentation of a charge-one spectrum and is not counted separately.

For an \(SU(3)\) representation \(R\), define
\begin{equation}
\tr_R F_3^2=A_R\,\tr F_3^2,\qquad
\tr_R F_3^4=C_R\bigl(\tr F_3^2\bigr)^2,\qquad
\tr_R F_3^3=E_R\,\tr F_3^3,
\end{equation}
with
\begin{center}
\begin{tabular}{c|cccc}
\toprule
\(R\) & \(\dim R\) & \(A_R\) & \(C_R\) & \(E_R\) \\
\midrule
\(\one\) & 1 & 0 & 0 & 0 \\
\(\three,\thbar\) & 3 & 1 & \(1/2\) & \(\pm1\) \\
\(\six,\sixbar\) & 6 & 5 & \(17/2\) & \(\pm7\) \\
\(\eight\) & 8 & 6 & 9 & 0 \\
\bottomrule
\end{tabular}
\end{center}

We label a physical full hypermultiplet
simply by \(R_q\), with signed charge \(q=-2,-1,0,1,2\), understanding
that its complex symplectic realization is
\begin{equation}
 R_q\oplus\overline R_{-q}.
 \label{eq:a2u1-signed-full-hyper}
\end{equation}
Thus \(R_{-q}\) denotes the opposite orientation.  The two summands in
\eqref{eq:a2u1-signed-full-hyper} are related by the symplectic
Majorana--Weyl reality condition, so a full hypermultiplet labelled by
\(R_q\) contributes \(\dim R\) to \(n_H\).

Let \(m_{R,q}\) be the multiplicity of the full hypermultiplet \(R_q\)
for \(R=\mathbf3,\mathbf6\).  For the charge magnitudes
\(r=1,2\), define
\begin{equation}
 M_{R,r}:=m_{R,r}+m_{R,-r},
 \qquad
 M_{R,0}:=m_{R,0}.
 \label{eq:a2u1-M-signed}
\end{equation}
For the self-conjugate \(A_2\) representations \(\mathbf1,\mathbf8\),
\(n_{1,r}\) and \(m_{8,r}\) count the charge-conjugate full-hyper blocks
at charge magnitude \(r\).  The total hypermultiplet dimension carried
by charge magnitude \(r\) is therefore
\begin{equation}
N_r:=n_{1,r}+3M_{3,r}+6M_{6,r}+8m_{8,r},
\qquad
\sum_{r=0}^2 N_r=256 .
\label{eq:a2u1-Nq}
\end{equation}

The \(A_2\) trace sums are
\begin{align}
\Sigma_A&=
\sum_{r=0}^2
\left(M_{3,r}+5M_{6,r}+6m_{8,r}\right),
\nonumber\\
2\Sigma_C&=
\sum_{r=0}^2
\left(M_{3,r}+17M_{6,r}+18m_{8,r}\right).
\label{SASC}
\end{align}
The coefficient of \(F\,\tr F_3^3\) must vanish separately.  In the
signed-charge notation this condition takes the particularly simple
form
\begin{equation}
{\cal C}_{\rm cubic}
=
\sum_{q=-2}^{2}q
\left(m_{3,q}+7m_{6,q}\right)
=0.
\label{cc}
\end{equation}
Reducing \eqref{cc} modulo two gives
\begin{equation}
\sum_{r=1}^2r\left(M_{3,r}+M_{6,r}\right)
\equiv0\pmod2 .
\label{mod2}
\end{equation}
No pairwise equality \(m_{R,r}=m_{R,-r}\) is imposed; cancellations
between different representations and charge sectors are retained.

With \eqref{cc} imposed, the anomaly polynomial takes the form
\begin{align}
I_8
=&
\bigl(\tr R^2\bigr)^2
+\frac{A_{\mathrm{adj}}-\Sigma_A}{6}\,
 \tr R^2\,\tr F_3^2
-\frac{S_2}{6}\,\tr R^2\,F^2
+\frac{2\Sigma_C-2C_{\mathrm{adj}}}{3}
 \bigl(\tr F_3^2\bigr)^2
\nonumber\\
&\hspace{7mm}
+4S_{Aq^2}\,(\tr F_3^2)F^2
+\frac{2S_4}{3}F^4
-12\,\tr F_3^2\,\tr_{\mathbf2}F_R^2
-16\bigl(\tr_{\mathbf2}F_R^2\bigr)^2,
\label{eq:I8-prefactorized-general}
\end{align}
where
\begin{equation}
S_2:=\sum_{r=0}^2r^2N_r,\qquad
S_4:=\sum_{r=0}^2r^4N_r,
\qquad
S_{Aq^2}:=
\sum_{r=0}^2r^2
\left(M_{3,r}+5M_{6,r}+6m_{8,r}\right).
\label{s24}
\end{equation}
Writing \(I_8\) in Green--Schwarz factorized form gives
\begin{align}
a\cdot a&=8,
&
a\cdot b_R&=0,
&
b_R\cdot b_R&=-8,
\label{av1}\\
a\cdot\bA&=\frac{6-\Sigma_A}{6},
&
\bA^2&=\frac{\Sigma_C-9}{3},
&
\bA\cdot b_R&=-3,
\label{eq:bA-bR}\\
a\cdot\bU&=-\frac{S_2}{6},
&
\bU^2&=\frac{S_4}{3},
&
\bA\cdot\bU&=S_{Aq^2}.
\label{eq:mixed}
\end{align}
Since no hyperino is charged simultaneously under the ordinary \(U(1)\)
and \(Sp(1)_R\),
\begin{equation}
b_{U(1)}\cdot b_R=0.
\label{ave}
\end{equation}

\subsection{Allowed abelian anomaly vectors}

For the local factorization analysis it is convenient to use
\begin{equation}
a=(-2,-2),
\qquad
b_R=(-2,2),
\label{conv2}
\end{equation}
with the off-diagonal bilinear form.  Equation~\eqref{ave} then implies
\begin{equation}
b_{U(1)}=(\beta,\beta).
\end{equation}
The first two equations in \eqref{eq:mixed} give
\begin{equation}
S_2=24\beta,
\qquad
S_4=6\beta^2.
\label{eq:S2S4-beta}
\end{equation}
Because \(|q|\le2\),
\begin{equation}
S_2=N_1+4N_2,
\qquad
S_4=N_1+16N_2,
\end{equation}
and hence
\begin{equation}
N_1=2\beta(16-\beta),
\qquad
N_2=\frac{\beta(\beta-4)}{2}.
\label{eq:N-beta}
\end{equation}
Integrality of \(N_1,N_2\) requires \(\beta\) to be even, while
nonnegativity gives \(0\le\beta\le16\) together with
\(\beta=0\) or \(\beta\ge4\).  In particular, \(\beta=2\) is excluded
because it gives \(N_2=-2\).  Thus, before using the mixed
\(A_2\)-\(U(1)\) constraint,
\begin{equation}
\beta\in\{0,4,6,8,10,12,14,16\},
\qquad
\ell:=\frac{\beta}{2}.
\label{beta}
\end{equation}

Now define
\begin{equation}
p:=a\cdot b_{A_2}.
\label{eq:def-p}
\end{equation}
Using \(b_{A_2}\cdot b_R=-3\) in the local basis
\eqref{conv2} gives
\begin{equation}
b_{A_2}
=
\left(
-\frac{p+3}{4},
\frac{3-p}{4}
\right).
\label{defp}
\end{equation}
Therefore
\begin{equation}
S_{Aq^2}
=
b_{A_2}\cdot b_{U(1)}
=
-\ell p.
\label{eq:SAq2-ellp}
\end{equation}
On the other hand, \eqref{mod2} and \eqref{s24} imply
\begin{equation}
S_{Aq^2}\equiv0\pmod2.
\label{eq:SAq2-even}
\end{equation}

The \(A_2\) representation coefficients satisfy
\begin{equation}
2C_R-A_R\equiv0\pmod{12}
\end{equation}
for every representation in \eqref{eq:a2u1-menu}.  Hence
\begin{equation}
2\Sigma_C-\Sigma_A\equiv0\pmod{12}.
\end{equation}
Using \eqref{eq:bA-bR} and \eqref{eq:def-p},
\begin{equation}
\Sigma_A=6(1-p),
\qquad
2\Sigma_C=18+\frac34(p^2-9),
\label{SASC2}
\end{equation}
so that
\begin{equation}
\frac34(p+1)(p+7)\equiv0\pmod{12}.
\end{equation}
Equivalently,
\begin{equation}
(p+1)(p+7)\equiv0\pmod{16},
\end{equation}
which gives
\begin{equation}
p\equiv\pm1\pmod8.
\label{eq:p-congruence}
\end{equation}
In particular, \(p\) is odd.  Combining
\eqref{eq:SAq2-ellp} and \eqref{eq:SAq2-even}, we then find
\(\ell\in2\mathbb Z\).  Therefore
\begin{equation}
\beta=4k,
\qquad
k\in\mathbb Z.
\label{eq:beta-4k}
\end{equation}
Thus the intermediate possibilities \(\beta=6,10,14\) are eliminated.
Combining \eqref{eq:beta-4k} with \eqref{beta} gives initially
\begin{equation}
k=0,1,2,3,4.
\end{equation}
For these values,
\begin{equation}
b_{U(1)}=-2ka=(4k,4k),
\qquad
S_2=96k,
\qquad
S_4=96k^2,
\end{equation}
and
\begin{equation}
N_1=32k(4-k),
\qquad
N_2=8k(k-1).
\label{eq:N-solutions}
\end{equation}
The \(k=4\) sector has
\((N_1,N_2)=(0,96)\), so every nonzero charge is \(|q|=2\).
It is therefore removed by the primitive normalization
\eqref{eq:primitive}.  The neutral and primitive sectors are consequently
\begin{equation}
k=0,1,2,3.
\label{eq:k-primitive}
\end{equation}
The \(k=0\) sector contains no charged matter and is the formal drone
sector; the no-drone scan retains \(k=1,2,3\).

\subsection{The seven \(A_2\) anomaly-vector classes}

Let
\begin{equation}
M_3:=\sum_{q=0}^2M_{3,q},
\qquad
M_6:=\sum_{q=0}^2M_{6,q},
\qquad
m_8:=\sum_{q=0}^2m_{8,q}.
\end{equation}
Subtracting the two equations in \eqref{SASC2} and using
\eqref{SASC} gives
\begin{equation}
M_6+m_8
=
\frac{(p+1)(p+7)}{16}.
\label{eq:M6m8-p}
\end{equation}
Moreover,
\begin{equation}
\Sigma_A=M_3+5M_6+6m_8
\ge5(M_6+m_8).
\end{equation}
Together with \(\Sigma_A=6(1-p)\), this gives
\begin{equation}
6(1-p)
\ge
\frac{5}{16}(p+1)(p+7).
\label{eq:p-bound}
\end{equation}
The inequality implies, for integral \(p\),
\begin{equation}
-27\le p\le0.
\end{equation}
Combining this finite interval with
\(p\equiv\pm1\pmod8\) gives precisely
\begin{equation}
p\in
\{-25,-23,-17,-15,-9,-7,-1\}.
\label{eq:p-seven}
\end{equation}
Thus the finiteness of the seven \(A_2\) classes follows directly from
the anomaly equations and nonnegativity of the representation
multiplicities.

The anomaly vectors $b_{A_2}$ corresponding to these values of $p$ are
\begin{equation}
 \bA\in
 \left\{
 \left(\frac{11}{2},7\right),
 \left(5,\frac{13}{2}\right),
 \left(\frac72,5\right),
 \left(3,\frac92\right),
 \left(\frac32,3\right),
 \left(1,\frac52\right),
 \left(-\frac12,1\right)
 \right\}.
\label{eq:seven-bA}
\end{equation}
Every value in \eqref{eq:p-seven} occurs for each sector \(k=0,1,2,3\).  Hence there are
\begin{equation}
7\times4=28
\end{equation}
formal anomaly-vector classes.  The seven \(k=0\) classes have a neutral
ordinary \(U(1)\) and are omitted by the drone-free condition, leaving
\(21\) no-drone classes.

\subsection{Unimodular embedding and positivity}

For every $p$ in \eqref{eq:p-seven}, we have
$p\equiv\pm1\pmod8$ and hence $p^2\equiv1\pmod{16}$.  Therefore
\begin{equation}
 b_{A_2}^2
 =2\left(-\frac{p+3}{4}\right)\left(\frac{3-p}{4}\right)
 =\frac{p^2-9}{8}\in2\mathbb Z+1.
\end{equation}
The charge lattice can be given directly in these coordinates.  For
each congruence class of $p$, take $\Lambda_p=\mathbb Z e_p\oplus
\mathbb Z f_p$ with basis
\begin{equation}
\begin{array}{c|cc}
 &e_p&f_p\\ \hline
 p\equiv 1\pmod8
   &\left(1,\frac12\right)&\left(-1,\frac12\right)\\
 p\equiv-1\pmod8
   &\left(\frac12,1\right)&\left(\frac12,-1\right)
\end{array}.
\end{equation}
Both bases satisfy $e_p^2=1$, $f_p^2=-1$, and $e_p\cdot f_p=0$.
Hence $\Lambda_p$ is odd and unimodular. It is then straightforward to verify that all the anomaly vectors lie in this odd unimodular lattice by checking that each has integral pairings with both basis vectors.
One can further show that $a$ is characteristic: for any
$x=m e_p+n f_p\in\Lambda_p$,
\begin{equation}
 a\cdot x-x^2=-3m+n-m^2+n^2\in2\mathbb Z,
\end{equation}
and $b_{U(1)}/2$ is a lattice element: 
\begin{equation}
\frac12 b_{U(1)}
=(2k,2k)
=3k e_p+k f_p
\in\Lambda_p.
\end{equation}

For the ghost-free domain, write $j=(j_1,j_2)$ and impose its unit-norm condition
$2j_1j_2=1$.  For $k=1,2,3$, positivity of the $U(1)$ and
$Sp(1)_R$ gauge kinetic terms requires
\begin{align}
 j\cdot b_{U(1)}&=4k(j_1+j_2)>0,\\
 j\cdot b_R&=2(j_1-j_2)>0.
\end{align}
The first inequality gives $j_1+j_2>0$, while the unit-norm condition
gives $j_1j_2>0$; together these imply $j_1,j_2>0$.  The second
inequality then gives $j_1>j_2$.  For every $p$ in
\eqref{eq:p-seven}, the remaining gauge kinetic coefficient is
automatically positive, since
\begin{equation}
 j\cdot b_{A_2}
 =\frac{(3-p)j_1-(p+3)j_2}{4}
 =\frac{3-p}{4}(j_1-j_2)-\frac p2j_2>0
 \qquad (p\le-1).
\end{equation}
Thus the three positivity conditions give precisely the common
ghost-free domain of all $21$ no-drone classes,
\begin{equation}
 2j_1j_2=1,\qquad j_1>j_2>0
 \qquad\Longleftrightarrow\qquad
 j=\frac1{\sqrt2}\left(e^\phi,e^{-\phi}\right),\quad\phi>0.
\end{equation}

\subsection{Enumeration and equivalence of spectra}

For each pair \((p,k)\), the scan solves the complete system of nonnegative-integer multiplicity equations, including \eqref{eq:a2u1-Nq}, \eqref{SASC}, \eqref{cc}, and \eqref{eq:mixed}. For \(k=1,2,3\), it additionally imposes the primitive charge condition \eqref{eq:primitive}, while \(k=0\) is retained only as a formal neutral sector.

The enumeration is then quotiented by the \(\mathbb Z_2\)
action that reverses the ordinary \(U(1)\) generator. In the
signed-charge notation this simply acts on the complex-representation
multiplicities as
\begin{equation}
m_{3,q}\longleftrightarrow m_{3,-q},
\qquad
m_{6,q}\longleftrightarrow m_{6,-q},
\qquad q=1,2,
\label{eq:a2u1-Z2}
\end{equation}
while the charge-paired \(\mathbf1\) and \(\mathbf8\) blocks are
unchanged. Equivalently, it exchanges the two possible orientations of each complex charged hypermultiplet.
Spectra related by this global transformation are counted only once~\footnote{Note that one first solves the anomaly equations with both orientations distinct, and quotients by global charge reversal only afterward.}.
The counts in Table~\ref{tab: A2U1} are
the numbers of \(\mathbb Z_2\)-orbits, not the numbers of oriented
solutions before this quotient.

The numbers of inequivalent spectra in the four formal \(k\)-sectors are
\begin{center}
\begin{tabular}{c|c|c|c|c}
\toprule
\(k\) & \(\bU\) & \((N_1,N_2)\) & inequivalent spectra & status\\
\midrule
0 & \((0,0)\) & \((0,0)\) & 68 & drone\\
1 & \((4,4)\) & \((96,0)\) & 813 & retained\\
2 & \((8,8)\) & \((128,16)\) & 22,408 & retained\\
3 & \((12,12)\) & \((96,48)\) & 155,317 & retained\\
\midrule
\multicolumn{3}{r}{No-drone total} & 178,538 & \\
\bottomrule
\end{tabular}
\captionof{table}{Numbers of inequivalent spectra in each
charge-dimension sector.  The \(k=0\) sector has no charged matter and
is excluded from the no-drone result.}
\label{tab: A2U1}
\end{center}

Thus, within the representation and charge domain
\eqref{eq:a2u1-menu}, there are \(178{,}538\) inequivalent primitive
no-drone spectra in \(21\) anomaly-vector classes.  Including the formal
neutral-\(U(1)\) sector gives \(178{,}606\) spectra in \(28\) formal
classes.  All retained classes satisfy local Green--Schwarz
factorization, the cubic \(A_2^3\)-\(U(1)\) condition, primitive charge
normalization, unimodular embedding in odd lattice
\(I_{1,1}\), and positivity of the gauge kinetic terms.

\section{Conclusions}
\label{sec:conclusions}

In this paper we have investigated anomaly constraints in six-dimensional
\(N=(1,0)\) supergravity with gauged \(Sp(1)_R\) symmetry and one tensor
multiplet.  The Green--Schwarz factorization conditions, together with the
requirement that the anomaly coefficients embed in a unimodular string charge
lattice, place strong arithmetic restrictions on the theory.  Remarkably, the unimodularity condition applied to the Gram matrix of \(a\) and \(b_R\) fixes the number of vector multiplets to one of only two positive possibilities, $n_V=12$ and $n_V=96$ \cite{Becker:2025xgy}. This already shows that the \(Sp(1)_R\)-gauged problem is much more rigid than the corresponding \(U(1)_R\)-gauged classification, and motivates an extensive search for consistent $Sp(1)_R$ gauged models.

In the \(n_V=96\) branch, and under the principal gauge-factor restriction that the non-\(R\) gauge algebra contains no \(U(1)\), \(A_1\) or \(A_2\) simple factors, we  carried out a complete scan over the resulting finite set of reductive gauge  sectors. The inclusion of $C_2$ already substantially increased the size of the search space, and the omission of $U(1), A_1$ and $A_2$ is due to the fact that their inclusion leads to an enormous enlargement of the search space. We find $4{,}094$ formal anomaly- and lattice-consistent spectra, but every one of them contains at least one drone. They are therefore excluded as physical models because the corresponding abelian gauge-kinetic matrix is degenerate \cite{Park:2011wv}, as reviewed in section~\ref{sec: consistency conditions}. 

In the \(n_V=12\) branch, if the gauge group is taken to be semisimple, then the non-\(R\)
gauge algebra is forced to be \(A_1^3\). In this case, and in the doublet/adjoint truncation 
described in the text, we find $4{,}337{,}331$ consistent models, grouped into $102$ classes. 
If we allow $U(1)$ factors, there are four possible gauge groups, of which the simplest one is $A_2\times U(1)\times Sp(1)_R$. In this case, we allow hyperfermions to be charged under $U(1)$, and for the representation set and charge range specified in the text, we found \(178{,}538\) solutions grouped into $21$ classes. 
 
Contrasting our results with those for the $U(1)_R$-gauged models, we note that searches performed under substantial restrictions on the group factors and representation menus have yielded remarkably few solutions, namely 33 in total~\cite{Randjbar-Daemi:1985tdc,Avramis:2005qt,Avramis:2005hc,Becker:2023zyb,Becker:2025xgy,Becker:2026etl}. However, a comprehensive search for \(U(1)_R\)-gauged models analogous to the enumeration of ungauged models in \cite{Hamada:2024oap} has not yet been carried out. 

The consistent solution space has been found to be enlarged considerably also in the case of diagonally gauged $U(1)_{R+}$  models \cite{Suzuki:2005vu,Guo:2026dvy} where a diagonal subgroup of $U(1)_R \subset Sp(1)_R$  and $U(1)\subset Sp(n_H)$ is gauged. It would be interesting to classify consistent models with diagonal $Sp(1)_{R+}$ gauging as well. 
Another extension of our work would be the study of the global form of the gauge group in all the consistent models found so far.

In determining the exact solutions of the anomaly-free minimal 6D supergravities, one typically considers the classical action, thereby neglecting the Green-Schwarz anomaly canceling counterterms. In that case, the $R$-symmetry ungauged supergravities admit $\mathrm{Minkowski}_6$ solutions, but the standard $U(1)_R$- or $Sp(1)_R$-gauged supergravities do not admit $\mathrm{Minkowski}_6$ or $\mathrm{(\mathrm{A})\mathrm{dS}}_6$ vacua. This is not pertinent to the question of local and global anomalies we have considered, but it may have consequences in a detailed study of anomaly inflows. A class of diagonally $U(1)_{R+}$-gauged models has been shown to admit not only $\mathrm{dS}_6$ solutions, but also $\mathrm{Minkowski}_6$ solutions \cite{Guo:2026dvy}. In the latter case, a subset of hyperscalars must be given specific nonvanishing expectation values. We expect this to continue to be the case in diagonally $Sp(1)_{R+}$ gauged models as well.

A central point of this paper is the considerable enlargement of the landscape of local and global anomaly-free $R$-symmetry gauged supergravity theories in minimal $6D$ supergravities. It is important to determine whether the
models found here satisfy additional consistency conditions suggested by anomaly inflow, and various Swampland criteria. The remarkable growth in the 
size of the solution space for consistent $R$-symmetry gauged models further motivates investigation of whether such models can be embedded in string/M/F theory.  

\subsection*{Data and code availability}

The representation data, Mathematica codes used to derive the single-factor anomaly-free spectra (with drones) in the $n_V=96$ branch, to impose the spectra-assembly method, and to reproduce the class-by-class enumeration of the \(A_1^3\times Sp(1)_R\), $A_2\times U(1)$ spectra, along with a PDF of the full $n_V=96$ branch catalogue are available as ancillary materials\cite{github} with the submission.

\subsubsection*{Acknowledgments}

We thank Yuji Tachikawa for helpful discussions. We acknowledge the use of ChatGPT in exploratory, computational, and editorial stages. 
This work is supported in part by the NSF grant PHY-2413006 and by endowment funds from the Mitchell Family Foundation.

\appendix

\section{Hypermultiplet multiplicity convention}
\label{app:hyper-multiplicity-conventions}

Throughout this paper we use a single physical convention.  For a
representation block
\begin{equation}
 {\cal R}_q=(R_1,\ldots,R_n)_q,
 \qquad
 d_{\cal R}:=\prod_{i=1}^n\dim R_i,
\end{equation}
we denote by \(m_{\cal R}\) its multiplicity measured in
\emph{full-hypermultiplet units}.  The relevant reality type is that of the
full product block \({\cal R}_q\), including any abelian charge.  The physical
reality rule is
\begin{equation}
\begin{array}{c|c}
\text{reality type of }{\cal R}_q & m_{\cal R} \\
\hline
\text{complex or real} & \mathbb Z_{\ge0} \\[1mm]
\text{pseudoreal} & \frac12\mathbb Z_{\ge0}
\end{array}
\label{tab:app-mult-dictionary}
\end{equation}
Thus \(m_{\cal R}=1\) denotes one full hypermultiplet in every case, while
\(m_{\cal R}=\tfrac12\) denotes one half-hypermultiplet and is allowed only
for a pseudoreal block.

For a complex block, one full hypermultiplet has complexified
representation
\begin{equation}
 {\cal R}_q\oplus\overline{\cal R}_{-q}.
\end{equation}
These two summands are the symplectic completion of a \emph{single} full
hypermultiplet.  Consequently the contribution of a block to the
hypermultiplet count is always
\begin{equation}
 \Delta n_H=m_{\cal R}\,d_{\cal R},
 \label{eq:app-grav-count}
\end{equation}
and the same physical multiplicity multiplies its contributions to the gauge
and mixed-anomaly traces.  In particular, there is no extra factor of two in
\eqref{eq:app-grav-count} for a complex representation.

A nonzero abelian charge makes the full block complex, so an isolated
half-hypermultiplet is not allowed for \(q\ne0\).  At \(q=0\), a
self-conjugate product block is real or pseudoreal according to its invariant
bilinear form.  For a tensor product of self-conjugate irreps, an odd number
of pseudoreal factors gives a pseudoreal product, whereas an even number gives
a real product.

For exact computer enumeration it is convenient, but purely auxiliary, to
integerize the physical multiplicity by defining
\begin{equation}
 x_{\cal R}:=2m_{\cal R}\in\mathbb Z_{\ge0}.
 \label{eq:app-doubled-variable}
\end{equation}
Then \(x_{\cal R}\) is even for complex or real blocks and may be odd for
pseudoreal blocks.  We refer to \(x_{\cal R}\) only as the
\emph{doubled} or \emph{integerized} multiplicity variable; it is not the
multiplicity displayed in a physical spectrum.  All spectra and tables in
this paper quote \(m_{\cal R}\) in full-hypermultiplet units.

\begingroup

\section{A concrete illustration of the assembly procedure}
\label{subsec:concrete-assembly-example}

We now illustrate the spectra--assembly method with an explicit example from the
catalogues.  The example is useful because the drone numbers can be followed at every step and
because the same lower-factor data lead first to a nontrivial two-factor solution and then to a
simple obstruction at three factors.  Recall that in the $n_V=96$ branch the number of drones is
\be
 M_{\rm drone}=93-\sum_i \dim G_i
  ,
\ee
since the gauged $Sp(1)_R$ contributes three vector multiplets.

We first choose a spectrum in the $A_4$ one-factor catalogue,
\begin{equation}
 \mathcal S_{A_4}
 =240\,\mathbf 1+16\,\mathbf 5+2\,\mathbf{10},
 \qquad
 b_{A_4}=(1,0),
 \qquad
 M_{A_4}=69,
 \label{eq:assembly-A4-shadow}
\end{equation}
and a spectrum in the $C_5$ catalogue,
\begin{equation}
 \mathcal S_{C_5}
 =64\,\mathbf 1+10\,\mathbf{10}+4\,\mathbf{44},
 \qquad
 b_{C_5}=\left(\frac32,1\right),
 \qquad
 M_{C_5}=38.
 \label{eq:assembly-C5-shadow}
\end{equation}
All displayed coefficients in this example use the full-hypermultiplet
convention.  In particular, the pseudoreal fundamental of $C_5$ appears
with physical multiplicity $10$ in \eqref{eq:assembly-C5-shadow}; this is
equivalent to twenty minimal half-hypermultiplets, but the latter number is
not used as a spectrum coefficient.

For $A_4\times C_5$ the vector budget fixes
\[
 M_{A_4C_5}=93-\dim A_4-\dim C_5=93-24-55=14.
\]
This is exactly compatible with the two chosen single-factor spectra: deleting $C_5$ replaces its $55$ vector
multiplets by drones and gives $14+55=69$ drones, while deleting $A_4$ gives $14+24=38$ drones.
Thus the projection lands on \eqref{eq:assembly-A4-shadow} and
\eqref{eq:assembly-C5-shadow}, respectively.

Let $m_{R,S}$ denote the physical multiplicity of a hypermultiplet in $(R,S)$.  The single-factor spectra multiplicities
already eliminate almost all possible links.  In particular,
\[
 m_{\mathbf5,\mathbf{44}}=m_{\mathbf{10},\mathbf{10}}
 =m_{\mathbf{10},\mathbf{44}}=0,
\]
because deleting the second factor would otherwise produce more than the available $16$ copies
of $\mathbf5$ or $2$ copies of $\mathbf{10}$ in the $A_4$ projection.  Hence the only possible
nontrivial link is
\[
 m\equiv m_{\mathbf5,\mathbf{10}}\in\{0,1\}.
\]
The mixed anomaly equation fixes it.  Since $\lambda_{A_4}=\lambda_{C_5}=1$ and the two
fundamentals have quadratic index one,
\begin{equation}
 b_{A_4}\!\cdot b_{C_5}
 =(1,0)\!\cdot\left(\frac32,1\right)
 =1
 =m,
 \label{eq:A4C5-mixed-example}
\end{equation}
so $m=1$.  The unlinked multiplicities are then fixed by the two single-factor spectra.
For the $A_4$ spectrum, the link $(\mathbf5,\mathbf{10})$ contributes
$\dim\mathbf{10}=10$ copies of $\mathbf5$, whereas for the $C_5$ spectrum it
contributes $\dim\mathbf5=5$ full-hypermultiplet units in the pseudoreal
$\mathbf{10}$.  Hence
\[
 16-10=6,\qquad
 2-0=2,\qquad
 10-5=5,\qquad
 4-0=4.
\]
Finally, the singlet multiplicity follows, for example, from the $A_4$ spectrum.
Each full-hypermultiplet unit in $(\mathbf1,\mathbf{10})$ contributes
$\dim\mathbf{10}=10$ neutral hypermultiplets after deleting $C_5$, so
\[
 s=240-5\times10-4\times44=14.
\]
The other single-factor spectrum gives the independent check
\[
 14+6\times5+2\times10=64.
\]
Therefore the two selected one-factor spectra assemble uniquely to
\begin{equation}
 \begin{aligned}
 \mathcal S_{A_4C_5}={}&14(\mathbf1,\mathbf1)
 +(\mathbf5,\mathbf{10})
 +6(\mathbf5,\mathbf1)
 +5(\mathbf1,\mathbf{10})\\
 &\quad+4(\mathbf1,\mathbf{44})
 +2(\mathbf{10},\mathbf1),
 \end{aligned}
 \label{eq:A4C5-assembled-spectrum}
\end{equation}
with
\[
 (b_{A_4},b_{C_5})=\left((1,0),\left(\frac32,1\right)\right),
 \qquad M_{\rm drone}=14.
\]

To see how the same construction works at the next step, choose the lattice-admissible $G_2$
one-factor spectrum,
\begin{equation}
 \mathcal S_{G_2}=96\,\mathbf1+31\,\mathbf7+\mathbf{27},
 \qquad
 b_{G_2}=(2,4),
 \qquad
 M_{G_2}=79.
 \label{eq:assembly-G2-shadow}
\end{equation}
As a useful intermediate check, assembling \eqref{eq:assembly-A4-shadow} with
\eqref{eq:assembly-G2-shadow} gives $M_{A_4G_2}=55$.  The only possible charged link is
$(\mathbf5,\mathbf7)$.  Since $\lambda_{A_4}=1$, $\lambda_{G_2}=2$ and both fundamentals have
quadratic index one, the mixed equation gives
\[
 b_{A_4}\!\cdot b_{G_2}=4=2m_{\mathbf5,\mathbf7},
 \qquad\text{hence}\qquad m_{\mathbf5,\mathbf7}=2.
\]
Then we have
\[
 \mathcal S_{A_4G_2}
 =66(\mathbf1,\mathbf1)+2(\mathbf5,\mathbf7)+2(\mathbf5,\mathbf1)
 +21(\mathbf1,\mathbf7)+(\mathbf1,\mathbf{27})+2(\mathbf{10},\mathbf1).
\]

Now consider a putative $A_4\times C_5\times G_2$ model extending
\eqref{eq:A4C5-assembled-spectrum} and \eqref{eq:assembly-G2-shadow}.  There are no drones,
\[
 M_{A_4C_5G_2}=93-24-55-14=0.
\]
Its three two-factor projections would have to be
\[
 \begin{array}{ccl}
 \text{delete }G_2: & A_4\times C_5, & M=14,\\
 \text{delete }C_5: & A_4\times G_2, & M=55,\\
 \text{delete }A_4: & C_5\times G_2, & M=24.
 \end{array}
\]
The first two are precisely the spectra constructed above.  The obstruction occurs in the third
one.  In the full-hypermultiplet convention the charged content of the above $C_5$ spectrum is
$10\,\mathbf{10}+4\,\mathbf{44}$, while the chosen $G_2$ spectrum has
$31\,\mathbf7+\mathbf{27}$.  A $\mathbf{44}$ of $C_5$ cannot link to a nontrivial $G_2$
representation because the smallest such representation has dimension $7>4$.  Likewise a
$(\mathbf{10},\mathbf{27})$ link is excluded already by the projected spectrum multiplicity.  Thus the only
possible mixed link is $m(\mathbf{10},\mathbf7)$.  Since
$(\mathbf{10},\mathbf7)$ is pseudoreal, the physical multiplicity may be
half-integral, $m\in\frac12\mathbb Z_{\ge0}$.  Deleting $G_2$ multiplies
this full-hypermultiplet multiplicity by seven, so the $C_5$ spectrum requires
\begin{equation}
 7m\le 10.
 \label{eq:C5G2-shadow-bound}
\end{equation}
The mixed anomaly equation is already written in physical full-hyper units:
\begin{equation}
 \begin{aligned}
 b_{C_5}\!\cdot b_{G_2}
 &=\left(\frac32,1\right)\!\cdot(2,4)
 =8\\
 &=\lambda_{C_5}\lambda_{G_2}
 m A_{\mathbf{10}}A_{\mathbf7}
 =2m,
 \end{aligned}
 \label{eq:C5G2-mixed-obstruction}
\end{equation}
so $m=4$, in direct contradiction with \eqref{eq:C5G2-shadow-bound}, which gives
$m\le 10/7$.

Consequently there is no compatible $C_5\times G_2$ spectrum for these one-factor data, and hence
there is no anomaly-free assembly of the form $A_4\times C_5\times G_2$, despite the fact that spectra of $A_4\times C_5$ and $A_4\times G_2$ exist.
\endgroup

\section{The $A_1^3 \times Sp(1)_R$ models}
\label{sec:tabulation of a1^3 models}

The complete list of the nine \(A_1^3\times Sp(1)_R\) spectra containing only singlets and doublets is tabulated below. All coefficients in the table are quoted in full-hypermultiplet units.  Hence a coefficient $1/2$ denotes one allowed half-hypermultiplet in a pseudoreal product representation. 

\begingroup
\scriptsize
\setlength{\LTleft}{0pt}
\setlength{\LTright}{0pt}
\setlength{\tabcolsep}{2pt}
\renewcommand{\arraystretch}{1.16}
\begin{longtable}{@{}>{\centering\arraybackslash}p{0\textwidth}
>{\centering\arraybackslash}p{0.225\textwidth}
>{\RaggedRight\arraybackslash}p{0.76\textwidth}@{}}
\toprule
No. & $(b_1,b_2,b_3)$ & $\rho^H$\\
\midrule
\endfirsthead
\toprule
No. & $(b_1,b_2,b_3)$ & $\rho^H$\\
\midrule
\endhead
\bottomrule
\endfoot
\bottomrule\\
\caption{\small Complete list of the nine
\(A_1^3\times Sp(1)_R\) spectra containing only singlets and doublets.}
\label{tab:A1-singlet-doublet-complete}\\
\endlastfoot
1 & $\bigl((0,1),(0,1),(0,1)\bigr)$ & $\begin{array}[t]{@{}l@{}} 160\,(\mathbf{1},\mathbf{1},\mathbf{1}) + 16\,(\mathbf{1},\mathbf{1},\mathbf{2}) + 16\,(\mathbf{1},\mathbf{2},\mathbf{1}) + 16\,(\mathbf{2},\mathbf{1},\mathbf{1}) \end{array}$ \\
2 & $\bigl((0,1),(0,1),(1,2)\bigr)$ & $\begin{array}[t]{@{}l@{}} 120\,(\mathbf{1},\mathbf{1},\mathbf{1}) + 36\,(\mathbf{1},\mathbf{1},\mathbf{2}) + 14\,(\mathbf{1},\mathbf{2},\mathbf{1}) + (\mathbf{1},\mathbf{2},\mathbf{2}) + 14\,(\mathbf{2},\mathbf{1},\mathbf{1}) + (\mathbf{2},\mathbf{1},\mathbf{2}) \end{array}$ \\
3 & $\bigl((0,1),(1,2),(1,2)\bigr)$ & $\begin{array}[t]{@{}l@{}} 84\,(\mathbf{1},\mathbf{1},\mathbf{1}) + 32\,(\mathbf{1},\mathbf{1},\mathbf{2}) + 32\,(\mathbf{1},\mathbf{2},\mathbf{1}) + 3\,(\mathbf{1},\mathbf{2},\mathbf{2}) + 14\,(\mathbf{2},\mathbf{1},\mathbf{1}) + \frac{1}{2}\,(\mathbf{2},\mathbf{2},\mathbf{2}) \end{array}$ \\
4 & $\bigl((0,1),(1,2),(1,2)\bigr)$ & $\begin{array}[t]{@{}l@{}} 88\,(\mathbf{1},\mathbf{1},\mathbf{1}) + 30\,(\mathbf{1},\mathbf{1},\mathbf{2}) + 30\,(\mathbf{1},\mathbf{2},\mathbf{1}) + 4\,(\mathbf{1},\mathbf{2},\mathbf{2}) + 12\,(\mathbf{2},\mathbf{1},\mathbf{1}) + (\mathbf{2},\mathbf{1},\mathbf{2}) + (\mathbf{2},\mathbf{2},\mathbf{1}) \end{array}$ \\
5 & $\bigl((1,2),(1,2),(1,2)\bigr)$ & $\begin{array}[t]{@{}l@{}} 48\,(\mathbf{1},\mathbf{1},\mathbf{1}) + 32\,(\mathbf{1},\mathbf{1},\mathbf{2}) + 32\,(\mathbf{1},\mathbf{2},\mathbf{1}) + 32\,(\mathbf{2},\mathbf{1},\mathbf{1}) + 2\,(\mathbf{2},\mathbf{2},\mathbf{2}) \end{array}$ \\
6 & $\bigl((1,2),(1,2),(1,2)\bigr)$ & $\begin{array}[t]{@{}l@{}} 52\,(\mathbf{1},\mathbf{1},\mathbf{1}) + 30\,(\mathbf{1},\mathbf{1},\mathbf{2}) + 30\,(\mathbf{1},\mathbf{2},\mathbf{1}) + (\mathbf{1},\mathbf{2},\mathbf{2}) + 30\,(\mathbf{2},\mathbf{1},\mathbf{1}) + \,(\mathbf{2},\mathbf{1},\mathbf{2}) + (\mathbf{2},\mathbf{2},\mathbf{1}) + \frac{3}{2}\,(\mathbf{2},\mathbf{2},\mathbf{2}) \end{array}$ \\
7 & $\bigl((1,2),(1,2),(1,2)\bigr)$ & $\begin{array}[t]{@{}l@{}} 56\,(\mathbf{1},\mathbf{1},\mathbf{1}) + 28\,(\mathbf{1},\mathbf{1},\mathbf{2}) + 28\,(\mathbf{1},\mathbf{2},\mathbf{1}) + 2\,(\mathbf{1},\mathbf{2},\mathbf{2}) + 28\,(\mathbf{2},\mathbf{1},\mathbf{1}) + 2\,(\mathbf{2},\mathbf{1},\mathbf{2}) + 2\,(\mathbf{2},\mathbf{2},\mathbf{1}) + (\mathbf{2},\mathbf{2},\mathbf{2}) \end{array}$ \\
8 & $\bigl((1,2),(1,2),(1,2)\bigr)$ & $\begin{array}[t]{@{}l@{}} 60\,(\mathbf{1},\mathbf{1},\mathbf{1}) + 26\,(\mathbf{1},\mathbf{1},\mathbf{2}) + 26\,(\mathbf{1},\mathbf{2},\mathbf{1}) + 3\,(\mathbf{1},\mathbf{2},\mathbf{2}) + 26\,(\mathbf{2},\mathbf{1},\mathbf{1}) + 3\,(\mathbf{2},\mathbf{1},\mathbf{2}) + 3\,(\mathbf{2},\mathbf{2},\mathbf{1}) + \frac{1}{2}\,(\mathbf{2},\mathbf{2},\mathbf{2}) \end{array}$ \\
9 & $\bigl((1,2),(1,2),(1,2)\bigr)$ & $\begin{array}[t]{@{}l@{}} 64\,(\mathbf{1},\mathbf{1},\mathbf{1}) + 24\,(\mathbf{1},\mathbf{1},\mathbf{2}) + 24\,(\mathbf{1},\mathbf{2},\mathbf{1}) + 4\,(\mathbf{1},\mathbf{2},\mathbf{2}) + 24\,(\mathbf{2},\mathbf{1},\mathbf{1}) + 4\,(\mathbf{2},\mathbf{1},\mathbf{2}) + 4\,(\mathbf{2},\mathbf{2},\mathbf{1}) \end{array}$ \\
\end{longtable}
\endgroup

We allow unrestricted multiplicities of the permitted blocks, requiring only that an individual product representation contain a triplet under at most one $A_1$ factor. This gives rise to a dramatically larger number of solutions, namely a total of $4{,}337{,}331$ models grouped into $102$ anomaly-vector classes.
For each nonempty class we select one canonical representative by minimizing,
in order,
\begin{equation}
 \left(s,\ \#\{\mu:m_\mu>0\},\ (m_1,\ldots,m_{19})\right),
 \label{eq:a1-representative-order}
\end{equation}
where the final comparison is lexicographic in the fixed representation order
$(d_1,d_2,d_3)\in\{1,2,3\}^3$, with $(1,1,1)$ removed and at most one entry
equal to $3$. With this selection rule, representative spectra for each one of the $102$ anomaly-vector classes are tabulated below.

\begingroup
\small
\setlength{\tabcolsep}{3pt}
\renewcommand{\arraystretch}{1.04}
\setlength{\LTleft}{0pt}
\setlength{\LTright}{0pt}
\begin{longtable}{@{}r c r r >{\raggedright\arraybackslash}p{0.64\textwidth}@{}}
\toprule
No. & class & $N$ & $s_{\min}$ & canonical representative \\
\midrule
\endfirsthead
\toprule
No. & class & $N$ & $s_{\min}$ & canonical representative \\
\midrule
\endhead
\midrule
\multicolumn{5}{r}{\emph{Continued on the next page}}\\
\endfoot
\bottomrule\\
\caption{Canonical representative spectra for the 102 nonempty
$A_1^3\times Sp(1)_R$ anomaly-vector classes under the full SMW rule.
Here $N$ is the number of inequivalent spectra in the indicated class,
after quotienting by the stabilizer of the class.
Of the 102 representatives, 65 have $s_{\min}=0$.
Light gray shading in the first two columns marks the four representatives
whose charged matter contains only doublets and no triplet block.
The coefficients in the last column use the full-hypermultiplet convention of appendix~\ref{app:hyper-multiplicity-conventions}; half-integer coefficients occur only for pseudoreal product representations.}
\label{tab:a1-full-smw-representatives}\\
\endlastfoot
\cellcolor{gray!20}1 & \cellcolor{gray!20}$(0,0,0)$ & 1 & 160 & $16\,(\mathbf{1},\mathbf{1},\mathbf{2}) \allowbreak{}+{}\allowbreak 16\,(\mathbf{1},\mathbf{2},\mathbf{1}) \allowbreak{}+{}\allowbreak 16\,(\mathbf{2},\mathbf{1},\mathbf{1}) \allowbreak{}+{}\allowbreak 160\,(\mathbf{1},\mathbf{1},\mathbf{1})$ \\
\cellcolor{gray!20}2 & \cellcolor{gray!20}$(0,0,1)$ & 1 & 120 & $36\,(\mathbf{1},\mathbf{1},\mathbf{2}) \allowbreak{}+{}\allowbreak 14\,(\mathbf{1},\mathbf{2},\mathbf{1}) \allowbreak{}+{}\allowbreak (\mathbf{1},\mathbf{2},\mathbf{2}) \allowbreak{}+{}\allowbreak 14\,(\mathbf{2},\mathbf{1},\mathbf{1}) \allowbreak{}+{}\allowbreak (\mathbf{2},\mathbf{1},\mathbf{2}) \allowbreak{}+{}\allowbreak 120\,(\mathbf{1},\mathbf{1},\mathbf{1})$ \\
3 & $(0,0,2)$ & 3 & 80 & $56\,(\mathbf{1},\mathbf{1},\mathbf{2}) \allowbreak{}+{}\allowbreak \frac{29}{2}\,(\mathbf{1},\mathbf{2},\mathbf{1}) \allowbreak{}+{}\allowbreak \frac{1}{2}\,(\mathbf{1},\mathbf{2},\mathbf{3}) \allowbreak{}+{}\allowbreak \frac{29}{2}\,(\mathbf{2},\mathbf{1},\mathbf{1}) \allowbreak{}+{}\allowbreak \frac{1}{2}\,(\mathbf{2},\mathbf{1},\mathbf{3}) \allowbreak{}+{}\allowbreak 80\,(\mathbf{1},\mathbf{1},\mathbf{1})$ \\
4 & $(0,0,3)$ & 3 & 60 & $60\,(\mathbf{1},\mathbf{1},\mathbf{2}) \allowbreak{}+{}\allowbreak 4\,(\mathbf{1},\mathbf{1},\mathbf{3}) \allowbreak{}+{}\allowbreak \frac{25}{2}\,(\mathbf{1},\mathbf{2},\mathbf{1}) \allowbreak{}+{}\allowbreak (\mathbf{1},\mathbf{2},\mathbf{2}) \allowbreak{}+{}\allowbreak \frac{1}{2}\,(\mathbf{1},\mathbf{2},\mathbf{3}) \allowbreak{}+{}\allowbreak \frac{25}{2}\,(\mathbf{2},\mathbf{1},\mathbf{1}) \allowbreak{}+{}\allowbreak (\mathbf{2},\mathbf{1},\mathbf{2}) \allowbreak{}+{}\allowbreak \frac{1}{2}\,(\mathbf{2},\mathbf{1},\mathbf{3}) \allowbreak{}+{}\allowbreak 60\,(\mathbf{1},\mathbf{1},\mathbf{1})$ \\
\end{longtable}
\endgroup

\section{The $A_2 \times U(1) \times Sp(1)_R$ models}
\label{sec:tabulation of a2xu(1) models}

We use
\(\eta=%
\) and \eq{conv2}.  A full
hypermultiplet in a complex representation \(R=\mathbf3,\mathbf6\) is
labelled simply by \(R_q\), with signed \(q\), as in
\eqref{eq:a2u1-signed-full-hyper}; its symplectic completion
\(R_q\oplus\overline R_{-q}\) is implicit.  Thus \(R_{-q}\) denotes the
opposite orientation, and one \(R_q\) contributes \(\dim R\) to \(n_H\).
The mixed cubic equation \(F\,\tr F_{A_2}^3=0\) is imposed explicitly.
The number of models in each class is quoted after the
\(\mathbb Z_2\) quotient \eqref{eq:a2u1-Z2}.  Within each class, the
displayed representative first minimizes the number of neutral
singlets, then the number of distinct occupied blocks.


\begingroup
\fontsize{9}{10}\selectfont
\setlength{\tabcolsep}{1.8pt}
\renewcommand{\arraystretch}{1.12}
\begin{longtable}{
@{}c c c c
>{\raggedright\arraybackslash}p{0.60\textwidth}
r@{}}
\label{tab:A2U1-corrected}\\
\toprule
\(k\)&Class&\(b_{A_2}\)&\(b_{U(1)}\)&Representative spectrum&
\# models\\
\midrule
\endfirsthead
\multicolumn{6}{c}{\tablename\ \thetable\ continued}\\
\toprule
\(k\)&Class&\(b_{A_2}\)&\(b_{U(1)}\)&Representative spectrum&
\# models\\
\midrule
\endhead
\midrule
\multicolumn{6}{r}{Continued on the next page}\\
\endfoot
\bottomrule\\
\caption{\small Representative classes of anomaly-free \(A_2\times U(1)\times Sp(1)_R\) models with
\(|q|\leq2\) and \(A_2\) representations restricted to \(\mathbf1,\mathbf3,\overline{\mathbf3},
\mathbf6,\overline{\mathbf6},\mathbf8\). There are $178{,}538$ inequivalent spectra 
in 21 anomaly-vector classes.}
\endlastfoot
1 & 8 & $\left(\frac{11}{2},7\right)$ & $(4,4)$ & \(36\,\mathbf 1_{1} + 5\,\mathbf6_{1} + 5\,\mathbf6_{-1} + 16\,\mathbf3_{0} + 12\,\mathbf6_{0} + 5\,\mathbf 8_{0}\) & 231 \\
1 & 9 & $\left(5,\frac{13}{2}\right)$ & $(4,4)$ & \(36\,\mathbf 1_{1} + \mathbf6_{1} + \mathbf6_{-1} + 6\,\mathbf 8_{1} + 20\,\mathbf3_{0} + 6\,\mathbf6_{0} + 8\,\mathbf 8_{0}\) & 419 \\
1 & 10 & $\left(\frac{7}{2},5\right)$ & $(4,4)$ & \(24\,\mathbf 1_{1} + 8\,\mathbf3_{1} + 8\,\mathbf3_{-1} + 3\,\mathbf 8_{1} + \mathbf 1_{0} + 39\,\mathbf3_{0} + 7\,\mathbf6_{0}\) & 96 \\
1 & 11 & $\left(3,\frac{9}{2}\right)$ & $(4,4)$ & \(34\,\mathbf 1_{1} + 14\,\mathbf3_{-1} + 2\,\mathbf6_{1} + \mathbf 8_{1} + 44\,\mathbf3_{0} + 2\,\mathbf6_{0} + 2\,\mathbf 8_{0}\) & 61 \\
1 & 12 & $\left(\frac{3}{2},3\right)$ & $(4,4)$ & \(51\,\mathbf 1_{1} + 3\,\mathbf3_{1} + 10\,\mathbf3_{-1} + \mathbf6_{1} + 34\,\mathbf 1_{0} + 42\,\mathbf3_{0}\) & 4 \\
1 & 13 & $\left(1,\frac{5}{2}\right)$ & $(4,4)$ & \(54\,\mathbf 1_{1} + 7\,\mathbf3_{1} + 7\,\mathbf3_{-1} + 58\,\mathbf 1_{0} + 34\,\mathbf3_{0}\) & 1 \\
1 & 14 & $\left(-\frac{1}{2},1\right)$ & $(4,4)$ & \(90\,\mathbf 1_{1} + \mathbf3_{1} + \mathbf3_{-1} + 130\,\mathbf 1_{0} + 10\,\mathbf3_{0}\) & 1 \\
2 & 15 & $\left(\frac{11}{2},7\right)$ & $(8,8)$ & \(2\,\mathbf 8_{2} + 42\,\mathbf 1_{1} + 5\,\mathbf3_{1} + 5\,\mathbf3_{-1} + 7\,\mathbf 8_{1} + 16\,\mathbf6_{0} + 2\,\mathbf 8_{0}\) & 4,796 \\
2 & 16 & $\left(5,\frac{13}{2}\right)$ & $(8,8)$ & \(2\,\mathbf 8_{2} + 36\,\mathbf 1_{1} + 10\,\mathbf3_{1} + 10\,\mathbf3_{-1} + 4\,\mathbf 8_{1} + 8\,\mathbf6_{0} + 8\,\mathbf 8_{0}\) & 13,471 \\
2 & 17 & $\left(\frac{7}{2},5\right)$ & $(8,8)$ & \(\mathbf 1_{2} + 3\,\mathbf3_{-2} + \mathbf6_{-2} + 29\,\mathbf 1_{1} + 22\,\mathbf3_{1} + 9\,\mathbf3_{-1} + \mathbf6_{1} + 16\,\mathbf3_{0} + 8\,\mathbf 8_{0}\) & 2,883 \\
2 & 18 & $\left(3,\frac{9}{2}\right)$ & $(8,8)$ & \(8\,\mathbf 1_{2} + \mathbf 8_{2} + 29\,\mathbf 1_{1} + 12\,\mathbf3_{1} + 19\,\mathbf3_{-1} + \mathbf6_{1} + 24\,\mathbf3_{0} + 5\,\mathbf 8_{0}\) & 1,198 \\
2 & 19 & $\left(\frac{3}{2},3\right)$ & $(8,8)$ & \(4\,\mathbf 1_{2} + 4\,\mathbf3_{-2} + 77\,\mathbf 1_{1} + 15\,\mathbf3_{1} + \mathbf6_{-1} + 4\,\mathbf 1_{0} + 36\,\mathbf3_{0}\) & 48 \\
2 & 20 & $\left(1,\frac{5}{2}\right)$ & $(8,8)$ & \(\mathbf 1_{2} + \mathbf3_{2} + 4\,\mathbf3_{-2} + 104\,\mathbf 1_{1} + 7\,\mathbf3_{1} + \mathbf3_{-1} + 7\,\mathbf 1_{0} + 35\,\mathbf3_{0}\) & 11 \\
2 & 21 & $\left(-\frac{1}{2},1\right)$ & $(8,8)$ & \(16\,\mathbf 1_{2} + 116\,\mathbf 1_{1} + 2\,\mathbf3_{1} + 2\,\mathbf3_{-1} + 88\,\mathbf 1_{0} + 8\,\mathbf3_{0}\) & 1 \\
3 & 22 & $\left(\frac{11}{2},7\right)$ & $(12,12)$ & \(8\,\mathbf3_{-2} + 4\,\mathbf6_{2} + 36\,\mathbf 1_{1} + 5\,\mathbf3_{1} + 3\,\mathbf3_{-1} + 6\,\mathbf6_{-1} + 12\,\mathbf6_{0} + 5\,\mathbf 8_{0}\) & 33,921 \\
3 & 23 & $\left(5,\frac{13}{2}\right)$ & $(12,12)$ & \(8\,\mathbf3_{-2} + 3\,\mathbf 8_{2} + 39\,\mathbf 1_{1} + 9\,\mathbf3_{1} + 3\,\mathbf6_{1} + 2\,\mathbf6_{-1} + 14\,\mathbf 8_{0}\) & 101,013 \\
3 & 24 & $\left(\frac{7}{2},5\right)$ & $(12,12)$ & \(16\,\mathbf3_{-2} + 27\,\mathbf 1_{1} + 5\,\mathbf3_{1} + 8\,\mathbf3_{-1} + 5\,\mathbf6_{1} + 24\,\mathbf3_{0} + 5\,\mathbf 8_{0}\) & 14,105 \\
3 & 25 & $\left(3,\frac{9}{2}\right)$ & $(12,12)$ & \(\mathbf3_{2} + 15\,\mathbf3_{-2} + 36\,\mathbf 1_{1} + 15\,\mathbf3_{1} + \mathbf3_{-1} + 2\,\mathbf6_{1} + 24\,\mathbf3_{0} + 5\,\mathbf 8_{0}\) & 6,117 \\
3 & 26 & $\left(\frac{3}{2},3\right)$ & $(12,12)$ & \(15\,\mathbf 1_{2} + 3\,\mathbf3_{2} + 8\,\mathbf3_{-2} + 66\,\mathbf 1_{1} + 10\,\mathbf3_{1} + 4\,\mathbf 1_{0} + 34\,\mathbf3_{0} + \mathbf6_{0}\) & 137 \\
3 & 27 & $\left(1,\frac{5}{2}\right)$ & $(12,12)$ & \(18\,\mathbf 1_{2} + 5\,\mathbf3_{2} + 5\,\mathbf3_{-2} + 90\,\mathbf 1_{1} + \mathbf3_{1} + \mathbf3_{-1} + 4\,\mathbf 1_{0} + 36\,\mathbf3_{0}\) & 22 \\
3 & 28 & $\left(-\frac{1}{2},1\right)$ & $(12,12)$ & \(45\,\mathbf 1_{2} + \mathbf3_{-2} + 90\,\mathbf 1_{1} + 2\,\mathbf3_{1} + 85\,\mathbf 1_{0} + 9\,\mathbf3_{0}\) & 2 \\
\end{longtable}
\endgroup

\bibliographystyle{jhep}

\bibliography{bibliography}

\end{document}